\documentclass[11pt]{article}
\usepackage[margin=1in]{geometry}

\usepackage[nokeyprefix]{refstyle}
\usepackage{varioref}
\usepackage{xr-hyper}
\usepackage{url}
\usepackage[hidelinks]{hyperref}
\usepackage{booktabs}
\usepackage{bbm}
\usepackage{color}
\usepackage{ifthen}
\usepackage{amsmath}
\usepackage{amssymb}
\usepackage{amsthm}
\usepackage{graphicx}
\usepackage{rotating}
\usepackage{float}
\usepackage{setspace}
\usepackage{fancyvrb}
\usepackage{dcolumn}
\usepackage{tabularx}
\usepackage{tikz}
\usetikzlibrary{positioning}
\usetikzlibrary{arrows.meta, matrix}
\usetikzlibrary{shapes,arrows,chains}
\usetikzlibrary{arrows,positioning,calc,topaths}
\usepackage{natbib}

\usepackage{longtable}
\usepackage{array}
\newcolumntype{L}[1]{>{\raggedright\let\newline\\\arraybackslash\hspace{0pt}}m{#1}}
\newcolumntype{C}[1]{>{\centering\let\newline\\\arraybackslash\hspace{0pt}}m{#1}}
\newcolumntype{R}[1]{>{\raggedleft\let\newline\\\arraybackslash\hspace{0pt}}m{#1}}

\newtheorem{theorem}{Theorem}
\newtheorem{lemma}{Lemma}

\newcommand*{\E}{\mathrm{E}}
\newcommand*{\Var}{\mathrm{Var}}
\newcommand*{\Prec}{\mathrm{Prec}}
\newcommand*{\Cov}{\mathrm{Cov}}
\newcommand*{\Cor}{\mathrm{Cor}}

\newcommand*{\N}{\mathcal{N}}

\renewcommand*{\vec}[1]{\boldsymbol{#1}}
\newcommand{\mat}[1]{\mathrm{\mathbf{#1}}}
\newcommand*\set[1]{\mathrm{#1}}
\newcommand*\col[1]{\mathcal{#1}}

\newcommand*{\affine}{\widebar{\vec\beta}}

\newcommand{\numcounties}{3,093}

\renewcommand*{\phi}{\varphi}
\renewcommand*{\epsilon}{\varepsilon}

\newcommand{\trans}{\intercal}
\newcommand{\onecov}{C}
\newcommand{\covs}{\bm \onecov}
\newcommand{\independent}{\perp\!\!\!\perp}

\newcolumntype{C}{>{\centering\arraybackslash}X}
\newcolumntype{d}[1]{D{.}{.}{#1}}
\newcolumntype{L}{D{.}{.}{1,2}}
\newcolumntype{P}[1]{>{\centering\arraybackslash}p{#1}}

\DeclareMathOperator*{\sign}{sign}

\DeclareFontFamily{U}{mathx}{\hyphenchar\font45}
\DeclareFontShape{U}{mathx}{m}{n}{
      <5> <6> <7> <8> <9> <10>
      <10.95> <12> <14.4> <17.28> <20.74> <24.88>
      mathx10
      }{}
\DeclareSymbolFont{mathx}{U}{mathx}{m}{n}
\DeclareFontSubstitution{U}{mathx}{m}{n}
\DeclareMathAccent{\widebar}{0}{mathx}{"73}

\title{Mitigating Unobserved Spatial Confounding when Estimating the Effect of Supermarket Access on Cardiovascular Disease Deaths}
\author{Patrick M Schnell$^{1,*}$, and
Georgia Papadogeorgou$^{2}$ \\[10pt]
\small $^1$Division of Biostatistics, College of Public Health, The Ohio State University, Columbus, OH, USA \\
\small $^2$Department of Statistical Science, Duke University, Durham NC, USA \\
\small $^*$schnell.31@osu.edu}
\date{}

\usepackage{bm}
\newcommand{\mb}{\bm}
\usepackage[capitalize,noabbrev]{cleveref}
\crefformat{equation}{(#2#1#3)}
\crefformat{section}{\S#1}

\newcommand{\XtXinv}{\left(\mat X^\trans \mat X\right)^{-1}}
\newcommand{\G}{\mat G}
\newcommand{\h}{\mat H}
\newcommand{\Q}{\mat Q}
\allowdisplaybreaks

\begin{document}
\tikzstyle{line} = [draw, -latex']

\maketitle

\begin{abstract}
Confounding by unmeasured spatial variables has received some attention in the spatial statistics and causal inference literatures, but concepts and approaches have remained largely separated.
In this paper, we aim to bridge these distinct strands of statistics by considering unmeasured spatial confounding within a causal inference framework, and estimating effects using outcome regression tools popular within the spatial literature.
First, we show how using spatially correlated random effects in the outcome model, an approach common among spatial statisticians, does not necessarily mitigate bias due to spatial confounding, a previously published but not universally known result.
Motivated by the bias term of commonly-used estimators, we propose an affine estimator which addresses this deficiency.
We discuss how unbiased estimation of causal parameters in the presence of unmeasured spatial confounding can only be achieved under an untestable set of assumptions which will often be application-specific. We provide a set of assumptions which describe how the exposure and outcome of interest relate to the unmeasured variables, and we show that this set of assumptions is sufficient for identification of the causal effect based on the observed data when spatial dependencies can be represented by a ring graph. We implement our method using a fully Bayesian approach applicable to any type of outcome variable.
This work is motivated by and used to estimate the effect of county-level limited access to supermarkets on the rate of cardiovascular disease deaths in the elderly across the whole continental United States.
Even though standard approaches return null or protective effects, our approach uncovers evidence of unobserved spatial confounding, and indicates that limited supermarket access has a harmful effect on cardiovascular mortality.
\end{abstract}


\section{Introduction}
\label{sec:introduction}

Over 17 million deaths were attributed to cardiovascular disease (CVD) worldwide in 2016, and the prevalence of CVD among people aged 20 or older in the United States that same year was 48\% \citep{aha2019heart}.
Poor nutrition and high body mass index are major risk factors of CVD \citep{aha2019heart}, and there is evidence that these risk factors are influenced by the availability of nearby supermarkets \citep{powell2007associations, laraia2004proximity}, which have historically had a higher prevalence of heart-healthy foods \citep{sallis1986san, pearce2008contextual}.
In a prospective cohort study of individuals who had undergone cardiac catheterization, living in an area (census tract) with low income and poor food access was associated with an increased risk of myocardial infarction or death \citep{kelli2019living}.

Here, our goal is to cast this question within a causal inference framework, and quantify the effect of county-level supermarket availability on the rate of CVD mortality among the elderly (65+ years old) in the United States. For that purpose, we compile a data set including mortality, store, and demographic information for all counties in the continental United States (see also \cref{sec:data}).
Demographic characteristics such as income have been previously associated with CVD risk factors \citep{kelli2017association}, and might be confounders of the effect of supermarket proximity and access on CVD mortality. Even though a number of demographic variables are included in our data set, the relationship of interest is possibly confounded by unobserved or difficult-to-define variables such as regional culture relating to personal vehicles, diet, and general health-consciousness, or state-level support for people with disabilities. Such variables could represent predictors of the exposure which influence where people live, whether or not they own a vehicle, or where businesses choose to locate, and could also represent predictors of the outcome such as how much people exercise, their stress levels, or what food they choose to eat regardless of supermarket availability.
Furthermore, these potentially unobserved demographic variables are expected to be spatially correlated, in that nearby counties are expected to have similar values of these unmeasured covariates.

Therefore, we are faced with estimating the causal effect of an exposure on an outcome using spatially-referenced, observational data, and under the threat of unmeasured spatial confounding. To address this challenge, we combine spatial statistics tools and causal inference methodology within a common framework. Even though some attention has been given to causal inference topics in the spatial statistics literature \citep{paciorek2010importance, Hodges2010, hughes2013dimension, hanks2015restricted}, and to spatial topics in the causal inference literature \citep{Verbitsky-savitz2012, Keele2015, papadogeorgou2016adjusting}, there is a substantial gap in the intersection of the two fields.

%Fitting regression models to spatially referenced data sets often results in spatially correlated residuals.
In classic spatial statistics, regression models are often augmented to include spatially correlated random effects in order to ``account'' or ``adjust'' for the spatial dependence in the outcome model residuals. However, there is substantial confusion about what exactly these spatial models are capable of accounting for \citep{hanks2015restricted}.
In some settings, spatial mixed models are employed to estimate the relationship between an exposure and outcome without conditioning on spatial information. In this context, \cite{Hodges2010} and \cite{hughes2013dimension} proposed including a spatial random effect that is orthogonal to the exposure of interest.
Other times, it is asserted that spatial models adjust for unobserved covariates which have a spatial dependence structure \citep{Congdon2013, Lee2015a}.
Nevertheless, the usual spatial models do not in general eliminate bias due to unobserved confounders, even when the residual variance components are known \citep{paciorek2010importance}.
Recently, \cite{keller2019selecting} discuss the interpretation of estimates from regression models that progressively include spatial basis functions of higher complexity, and they conclude that increasing adjustment might even lead to bias amplification.

From a different perspective, causal inference methodology with spatial data and in the presence of unmeasured spatial confounding has been quite limited, and, to our knowledge, it has been restricted to classic causal inference tools. Within a regression discontinuity framework,  \cite{Keele2015} match treated to control units  separated by a boundary minimizing geographical distance of matched pairs and balancing observed covariates. Relatedly, \cite{papadogeorgou2016adjusting} proposed matching treated to control units on a criterion incorporating both propensity scores and geographical distance. Although these approaches can, in some cases, address the problem of interest to spatial statisticians, they are not immediately compatible with models commonly used in spatial data analysis which are most often grounded in outcome regression. An exception is found in \cite{thaden2018structural} where the authors propose a structural equation modeling approach treating the spatial variable as a confounder in a geoadditive model in order to eliminate bias from the unmeasured spatial variable.

In the causal inference literature, unmeasured confounding has been most often dealt with in the realm of sensitivity analysis. Sensitivity analysis is a powerful approach which aims to quantify the robustness of estimated effects to different amounts of unmeasured confounding \citep{rosenbaum1983assessing, Rosenbaum2002, imbens_rubin_2015, VanDerWeele2017}. However, sensitivity analysis does not directly adjust effect estimates for the presence of such confounders, which is what the methodology presented here and the works referenced above aim to achieve.

In this paper, we seek to bolster the bridge between spatial data analysis and causal inference. In order to do so, we consider unmeasured confounding within a formal causal inference framework and examine estimation approaches grounded on models and tools often employed by spatial statisticians.
%In order to provide a representation clear to as many readers as possible, we provide information that might be obvious to the experts of each field.
We start by focusing on continuous outcomes and linear models, studying the bias of commonly-used estimators. We propose a model-based approach to estimate the effect of a change in the exposure on an outcome of interest in the presence of unmeasured spatial confounding. Our approach is designed to easily incorporate popular tools in spatial statistics such as hierarchical and linear mixed models.
We explain that identification and estimation of the causal parameter in the presence of unmeasured confounding requires untestable assumptions regarding the unmeasured confounders and their relationship with the treatment and outcome of interest.
In general, these assumptions have to be application-specific, identification of the causal parameter needs to be evaluated separately for each set of assumptions, and the proposed estimator would have to be adapted to alternative identifying assumptions.
For continuous treatments (referred to as \textit{exposures}), we provide one set of assumptions that is sufficient for identification of the causal exposure-response curve, and one that is not. Importantly, our results illustrate that, when spatial dependencies can be represented using a ring graph, the components in our estimator involving the unmeasured confounder can be identified based solely on spatial dependencies in the observed data.
Within the context of our study, we extend our approach to non-continuous outcomes and generalized linear mixed models employing a fully-Bayesian approach, and we carefully discuss the plausibility of the causal assumptions. 
While our development is in the context of areal data, refinements in the context of point-referenced data are possible and are discussed where applicable.

In \cref{sec:data} we present a detailed description of our data set and present preliminary analyses using non-spatial and commonly-used spatial regression models that yield suspect results.
In \cref{sec:causal}, we define the causal estimand in terms of potential outcomes for continuous exposures and discuss commonly invoked identifiability assumptions when the observed covariates include a sufficient confounding adjustment set.
The proposed methodology is introduced in \cref{sec:estimators} within the context of linear models. There, we re-derive the result by \cite{paciorek2010importance} stating that commonly-used spatial regression models do not recover the estimands of interest in the presence of unmeasured spatial confounding. Motivated by the bias of commonly used estimators, we propose the affine estimator, and we provide a set of assumptions relating the exposure and outcome to the unmeasured variables based on which the causal parameter is identifiable from observed data.
In \cref{subsec:affine_bayesian}, we extend the affine estimator in the context of non-linear models and suggest using a Bayesian approach.
The estimator is compared to the currently-used estimators under various generative mechanisms via simulation in \cref{sec:simulation}.
In \cref{sec:example} we discuss the plausibility of our assumptions within the context of our study, and we use the affine estimator to estimate the county-level effect of poor supermarket availability on CVD mortality. Our study illustrates the potential of the affine estimator in mitigating bias from unmeasured spatial confounders, returning effect estimates that are qualitatively different from the ones in \cref{sec:data}, and more in line with subject-matter knowledge. We conclude with a discussion in \cref{sec:discussion}.

\section{County-level supermarket availability and CVD mortality}
\label{sec:data}

We compile a data set including mortality, store availability, demographic and behavioral data on \numcounties{} out of 3,109 counties and county equivalents in the 48 contiguous states and the District of Columbia.
For each areal unit (county, or county equivalent), supermarket availability is defined as the proportion of housing units during 2006 that are more than 1 mile from the nearest supermarket or large grocery store and do not have a car, obtained from the Food Environment Atlas, June 2012 release \citep{usda2012food}.
County-level population and cardiovascular disease mortality count (ICD-10 codes I00--I99) during 2007 for residents 65 years old and over were obtained from the United States Centers for Disease Control and Prevention (CDC) WONDER query system \citep{cdc2018underlying}. 
Due to privacy constraints, county death counts below 10 are censored by CDC WONDER.
Figure~\ref{fig:maps} displays the exposure and crude outcome relative risk, without covariate-assisted estimation or smoothing. Demographic information was acquired from the 2000 Census and includes, among others, information on urbanicity, poverty, and population mobility. Covariate information also includes estimates of 2006 smoking rates derived from CDC Behavioral Risk Factor Surveillance System data \citep{dwyer2014cigarette}.
Appendix~\ref{app_sec:data} provides additional information on the data sources, data collection, and data processing pipeline, including links to the publicly-available data sets, and a table including names and descriptive statistics of available covariates.

\begin{figure}[!b]
    \centering
    \includegraphics[scale = 0.9]{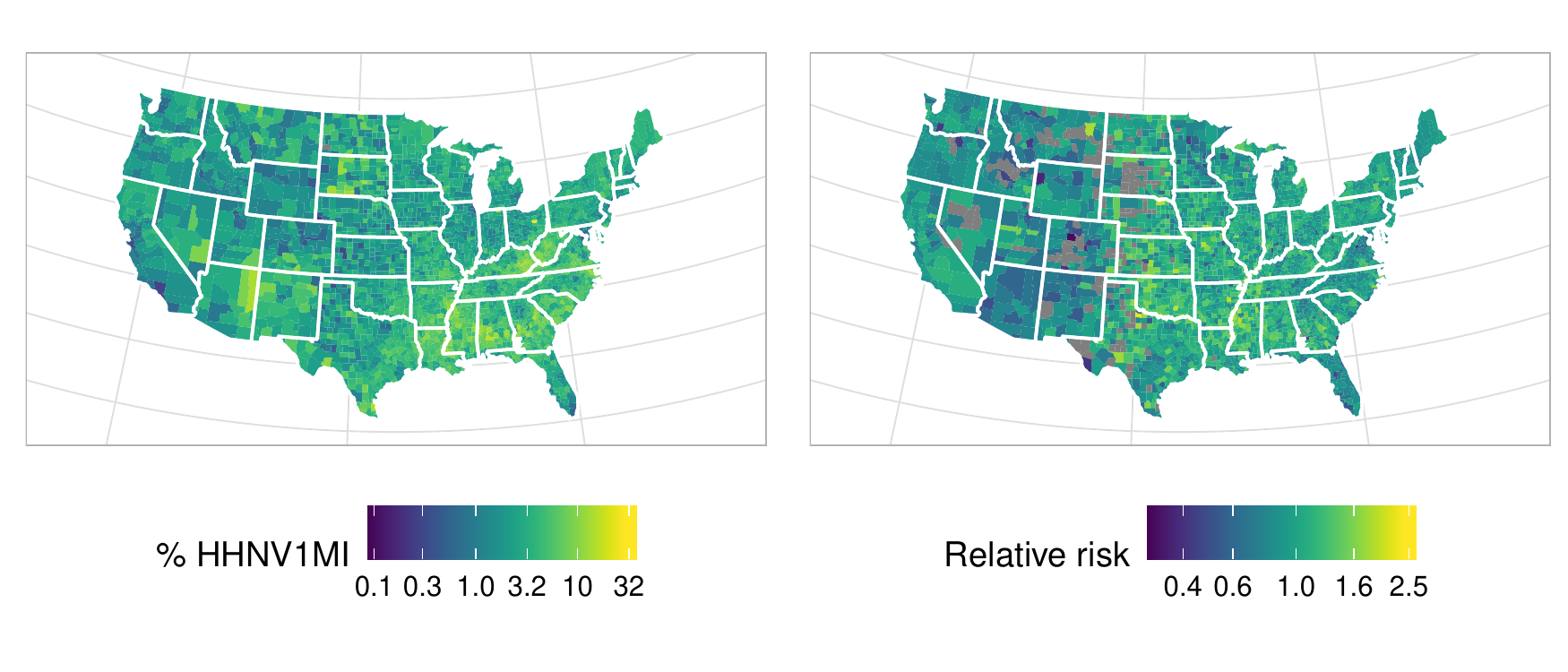}
    \caption{Percent of households with no vehicle and more than 1 mile from a supermarket or large grocery store (left, \% HHNV1MI), and observed relative risk of CVD deaths in the 65+ age range (right, relative risk).}
    \label{fig:maps}
\end{figure}

At this point, we consider two common analyses investigating the relationship between limited county-level access to supermarkets on cardiovascular deaths in the elderly. For the first analysis, we model CVD mortality counts as a Poisson-distributed outcome using a log link with the exposure and all covariates listed in \cref{app_tab:table1} as predictors.
Internal standardization was implemented by using log expected death count as the offset: the population age 65+ in each county was multiplied by the overall crude rate in the same age range.
All covariates were standardized, and the model was fit within the Bayesian paradigm under Gaussian priors with mean 0 and standard deviation 10 on regression coefficients.
Results are presented on original scales unless otherwise noted.
Samples from the posterior distribution were obtained via a Gibbs sampler, with a Metropolis-Hastings block update for all regression coefficients.
The Gibbs sampler was run for 10,000 iterations after a 1,000-iteration burn-in.
Censored outcomes were imputed subject to the known privacy constraint.

The second analysis we implemented is a common analysis method for areal spatial data. We included a spatially correlated random effect for county, $\vec{U}$, in the linear predictor, and assumed it follows a conditional autoregressive (CAR) structure \citep{besag1974spatial}, i.e., 
\begin{equation}
  U_i | \vec{U}_{-i} \sim \mathcal{N}\left[ \phi_U \sum_{j \in \partial_i} U_j / |\partial_i|, (\tau_U |\partial_i|)^{-1}\right],
\end{equation}
where $\partial_i$ and $|\partial_i|$ are the set and number of $i$'s neighbors, respectively.
A multivariate normal representation of the distribution of $\vec{U}$ is then available as
\begin{equation}
  \vec{U} \sim \mathcal{N}\left[\vec{0}, \tau_U^{-1} (\mat{D} - \phi_U \mat{W})^{-1} \right],
  \label{eq:multivariate_CAR}
\end{equation}
where $w_{ij} = 1$ if $i$ and $j$ are neighbors, and 0 otherwise, and $\mat{D}$ is diagonal with entries $|\partial_i|$ \citep{banerjee2004hierarchical}.
A uniform prior on $(-1, 1)$ was used for $\phi_U$ and a gamma prior with shape and rate parameters equal to 5 was used for $\tau_U$, jointly restricted to require the precision matrix of $\vec{U}$ to be positive definite.

Based on these two models, we acquired what the estimated effect of a 1 percentage point increase in households with poor supermarket access on cardiovascular mortality would be if each model was specified correctly and was sufficiently adjusted for confounding variables.
The non-spatial estimate (analysis 1) indicated that increasing poor supermarket availability is protective of CVD deaths, with an estimated relative risk of 0.968 (95\% CI 0.962 to 0.973). A protective effect of poor supermarket access is not consistent with either theoretical or data-driven understanding of the phenomenon, enforcing our belief that the result is at least partially due to unobserved, or poorly adjusted-for confounders.
The estimate from the spatial model (analysis 2) is effectively null with an estimated relative risk of 0.999 (95\% CI 0.988 to 1.011).
Although the spatial estimate differs from the non-spatial estimate in both location and credible interval width, we show in the next section that the spatial estimate does not necessarily mitigate bias from unobserved confounders.

\section{Causal estimands and classic identifiability assumptions}
\label{sec:causal}

Broadly speaking, the causal inference literature places substantial emphasis on defining target quantities of interest, referred to as \emph{estimands}, and determining sufficient assumptions under which such estimands (which include unobservable quantities) are identifiable based on the observed data.
We begin by defining estimands of interest following the potential outcome framework formalized by \cite{Rubin1974} and extended to continuous exposures by \cite{hirano2004propensity}.
A necessary condition for an exposure $Z$ to have an effect on an outcome $Y$ is that $Z$ is temporally precedent.
We make the stable unit treatment value assumption (SUTVA, \cite{Rubin1980}) which states that there is a single version of each treatment level and there is no interference between units.
Based on SUTVA, we can use $Y_i(z)$ to represent the value that would have been observed at location $i$ had it received exposure $z \in \col{Z}$, where $\col Z$ includes all possible values of the continuous $Z$, and $i = 1, 2, \dots, n$.
Then, $Y_i(z)$ is the \emph{potential outcome} for location $i$ at exposure level $z$, and unit $i$'s observed outcome $Y_i$ is the potential outcome for the observed level of the treatment, $Y_i = Y_i(Z_i)$.

The most common estimands for continuous treatments are the \emph{population average exposure-response curve} (PAERC) $\mu(z) = \E[Y_i(z)], z \in \col Z$, and the expected rate of change in the outcome for an infintensimal change in the exposure around $z$, $\mu'(z)$. Since $\mu(z)$ represents the average outcome value over the whole population had {\it everyone} experienced exposure $z$, it is clear that $\mu(z)$ includes unobserved quantities, and assumptions need to be made to ensure identifiability and to estimate it using data.
%Identifiability of an estimand means that, even though it is defined in terms of potential outcomes many of which are not observed, it can be expressed as a function of only the observed data.
The positivity and no unmeasured confounding assumptions (referred to together as the ignorability assumption) form a sufficient set of assumptions for identifiability of $\mu(z)$. Positivity states that all units can experience any $z \in \mathcal Z$, and the no unmeasured confounding assumption states that there exist \textit{measured} covariates $\covs$ which satisfy that, conditional on $\covs$, the observed exposure $Z$ is independent of the potential outcomes $Y(z)$, denoted as $Z \independent Y(z) | \covs, z \in \col{Z}$. (See Appendix~\ref{app_sec:causal} for a discussion on identifiability of $\mu(z)$ based on these assumptions.)

Confounders $\covs$ are generally thought of as temporally precedent to the exposure $Z$ and as common predictors of $Z$ and $Y$, as shown in \cref{fig:causal-diagram}. Since temporal order of variables matters in drawing causal conclusions, observed data are conceived as if generated in the following order: $[\covs]$, $[Z|\covs]$ and $[Y | Z, \covs]$. If the identifiability conditions of positivity and no unmeasured confounding are met in the observed data, estimation can proceed by imitating the data generating mechanism for the exposure $Z | \covs$, known as the propensity score \citep{rosenbaum1983central}, the data generating mechanism for the outcome $Y | Z, \covs$ \citep[e.g.,][]{hill2011bayesian, hahn2018regularization}, or both \citep{robins1995semiparametric, hahn1998role, zigler2014uncertainty, zhou2019penalized}.
In order to adhere to common approaches of spatial statistics which emphasize analytical models imitating the outcome generative model, our primary focus is on modeling $Y | Z, \covs$.

\begin{figure}[!t]
\centering
	\large{\begin{tikzpicture}[%
		->,
		>=stealth,
		node distance=1cm,
		pil/.style={
			->,
			thick,
			shorten =2pt,}
		]
		\node (1) {$\covs$};
		\node[right=of 1] (2) {$Z$};
		\node[right=of 2] (3) {$Y$};
		\draw [->] (1.east) -- (2.west);
		\draw [->] (2.east) -- (3.west);
	    \draw [->] (1) to [out=30, in=150] (3);
	\end{tikzpicture}}
	\caption{Assumed causal diagram for the generative model. The vector \texorpdfstring{$C$}{$C$} may represent a collection of multiple confounders.}
  \label{fig:causal-diagram}
\end{figure}
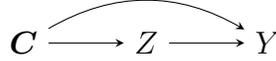

Even though confounding adjustment is necessary to draw causal conclusions, $\covs$ might include components that are not measured, hence violating the no unmeasured confounding assumption.
Denote $\covs = (\covs^m, \covs^u)$ representing the measured and unmeasured components, respectively.
At this point, we assume that at least some of the variables in $\covs^u$ are spatial and refer to \cref{subsec:identifiability} for a further discussion on this requirement.
We refer to a variable as ``spatial'' if the correlation of the variable for two observations depends on their geographic locations. For areal data like the ones in our study, this could refer to adjacency of counties.
For point referenced data, it could refer to the geographical distance of two points.

In this section and the next, we discuss unobserved spatial confounders in the case of continuous outcomes and linear models. Focusing on this setting allows for straightforward application of theory from least squares estimation of regression coefficients and restricted maximum likelihood estimation of variance parameters. We return to non-Gaussian outcomes for the simulation study and data analysis, where we employ a fully Bayesian approach. Assume here that potential outcomes arise in the following manner:
\begin{equation}
  \label{eq:dgm-general}
  Y_i(z) = \eta(z, \covs_i^m) + g(\covs_i^u) + \epsilon_i(z),
\end{equation}
for some function $\eta$, and $\epsilon_i(z)$ a mean zero random variable and independent of $\covs$.
In \cref{eq:dgm-general}, $\covs^u$ is assumed to not interact with $Z$ and $\covs^m$.  We denote $U = g(\covs^u)$, representing the cumulative contribution of all unobserved covariates. Since at least some components of $\covs^u$ are spatial, $U$ also has a spatial correlation structure. Without loss of generality, we may assume $\E[g(\covs^U)] = 0$ by absorbing any non-zero mean into the the intercept in $\eta(z, \covs^m$).

\section{The affine estimator in linear models: Addressing omitted variable bias of classic estimators}
\label{sec:estimators}

In this section, we discuss how the classic approaches to estimation within the spatial statistics literature are biased for estimating causal quantities in the presence of unmeasured confounders, and in the context of linear regression. The bias results derived below are in line with results in \cite{paciorek2010importance}, and they motivate the affine estimator, which is designed to explicitly remove the bias of the existing estimators. An extension to non-linear models within the Bayesian framework is presented in \cref{subsec:affine_bayesian}.

\subsection{Omitted variable bias of ordinary and generalized least squares estimators}

Let $\mat{X} = (\bm 1 \ | \ \bm Z \ | \ \covs^m)$ be the design matrix containing an intercept, the exposure $\vec{Z}$, and measured covariates, and let $\mat{X}_{(-z)}$ be the design matrix including an intercept and measured confounders, but not the exposure $\vec{Z}$. For simplicity of presentation, we assume that the causal exposure response curve is linear, $\eta(z, \covs_i^m) = \vec{\beta}_{(-z)}^\trans \vec{x}_{i,(-z)}  + \beta_z z$, which will be relaxed in \cref{sec:semiparametric}. This simplification implies that $\mu'(z) = \beta_z$, corresponding to the usual linear regression coefficient targeted in the spatial statistics literature.
Using vector notation let $\vec Y = (Y_1, Y_2, \dots, Y_n)^\trans$, with $\vec Z, \vec U, \vec \epsilon$ defined analogously. Then, {\it if} all of $\vec Y, \vec Z, \vec U$, and $\mat{X}_{(-z)}$ were observed, estimation of the regression model
\begin{equation}
  \label{eq:dgm}
  \vec{Y} = \mat{X}_{(-z)} \vec{\beta}_{(-z)} + \vec{Z} \beta_z + \vec{U} + \vec{\epsilon},
\end{equation}
would lead to consistent estimation of the causal effect through estimation of $\beta_z$.

However, the above model cannot be directly used in settings where $\vec{U}$ is not measured. In \cref{eq:dgm}, $\vec{Z}$ and $\vec{U}$ are correlated, but $\vec{\epsilon}$ is independent of $(\vec{Z}, \vec{U})$.
Thus $\vec{U}$ and $\vec{\epsilon}$ is a partition of the variability in $\vec{Y}$ not due to $\vec{Z}$ into one component ($\vec{\epsilon}$) which is independent of $\vec{Z}$ and one ($\vec{U}$) which is not. 
If $\vec{U}$ is correlated with $\vec{Z}$ and is omitted from the outcome regression, the \emph{ordinary least squares} (OLS) estimator of $\vec \beta = (\vec \beta_{(-z)}^\trans, \beta_z)$, $\widehat{\vec \beta}$, will be biased.  This is evident by examining the conditional expectation of $\widehat{\vec \beta}$:
\begin{equation}
\label{eq:ev-beta-hat}
\begin{aligned}
\E\big(\widehat{\vec{\beta}} | \mat{X}\big)
&= \E \big[ \XtXinv \mat{X}^\trans \vec{Y} | \mat{X} \big]
= \vec{\beta} +\XtXinv \mat{X}^\trans \E(\vec{U} | \mat{X}).
\end{aligned}
\end{equation}
Considering $\E\big(\widehat{\vec \beta}\big) = \E\big[\E\big(\widehat{\vec \beta} | \mat{X} \big) \big]$, we see that $\widehat{\vec \beta}$ will be biased for $\vec\beta$ since the second term will, in general, be non-zero for correlated $\vec{U}, \vec{Z}$.%, and that the bias conditional on $\vec{Z}$ will not generally disappear when marginalizing over $\vec{Z}$.

When $\vec{U}$ is omitted from the regression model, the component of $\vec{U}$ not attributed to $\mat{X}$ will be incorporated in the residuals.
Since $\vec{U}$ is spatially structured, residuals of the regression of $\vec{Y}$ on $\mat{X}$ will also be spatially correlated.
In an effort to account for residual spatial correlation, spatial linear mixed models are often adopted.
Typically, such models represent mechanisms similar in form to \cref{eq:dgm}, but in which all right-hand-side variables are \textit{assumed} to be independent, and some assumptions are made about the form of $\Var[\vec{U}] = \Var[\vec U | \mat{X}]$.
These models aim to explain the spatial correlation in the residuals and they are often effective at improving efficiency. However, they do not necessarily alleviate the omitted variable bias \citep{paciorek2010importance}.
If $\Var[\vec{Y} | \mat{X}]$ (which depends on $\Var[\vec{U}]$) is known, the \emph{generalized least squares} (GLS) estimator of $\vec{\beta}$ is 
\begin{equation}
  \label{eq:est-lmm}
  \widetilde{\vec{\beta}} = \{ \mat{X}^\trans (\Var[\vec{Y} | \mat{X}])^{-1} \mat{X} \}^{-1} \mat{X}^\trans (\Var[\vec{Y} | \mat{X}])^{-1} \vec{Y}.
\end{equation}
with conditional expected value
\begin{equation}
  \label{eq:est-lmm-ev}
  \E\big( \widetilde{\vec{\beta}} | \mat{X}\big) = \vec{\beta} +  \{ \mat{X}^\trans (\Var[\vec{Y} | \vec{Z}])^{-1}\mat{X} \}^{-1} \mat{X}^\trans (\Var[\vec{Y} | \mat{X}])^{-1} \E[\vec{U} | \mat{X}].
\end{equation}
Therefore, even if $\Var[\vec{Y} | \mat{X}]$ is known,  $\widetilde{\vec{\beta}}$ remains biased.
This result indicates that including a spatial random effect in the regression model does not necessarily mitigate or eliminate bias arising from unmeasured spatial confounders.

\subsection{The affine estimator to account for omitted spatial variables}
\label{sec:our_estimator}

The results presented above establish that spatial correlation in the outcome model residuals might arise due to spatial predictors of $\vec{Y}$, and commonly used approaches to estimate $\beta_1 = \mu'(z)$ are biased in the presence of unmeasured confounding by a spatial variable $\vec{U}$.
It is now clear that mitigating bias from unmeasured spatial variables cannot be achieved based solely on an outcome regression model without making additional assumptions, nor by harvesting spatial information found solely in the outcome model residuals.

An investigation of the formulas in \cref{eq:ev-beta-hat} and \cref{eq:est-lmm-ev} shows that bias of both least squares estimators arises from the non-zero correlation between $\vec{U}$ and $\vec{Z}$, leading to a non-zero $\E[\vec U | \mat X]$. Inspired by the form of the bias, we propose an estimator that includes a component that depends on $U$. Consider the \emph{affine estimator}:
\begin{equation}
  \label{eq:est-joint}
  \widebar{\vec{\beta}} = \{ \mat{X}^\trans (\Var[\vec{Y} | \mat{X}])^{-1} \mat{X} \}^{-1} \mat{X}^\trans (\Var[\vec{Y} | \mat{X}])^{-1} \{ \vec{Y} - \E[\vec{U} | \mat{X}] \},
\end{equation}
which replaces $\vec{Y}$ by $\vec{Y} - \E[\vec{U} | \mat{X}]$ and is unbiased if $\E[\vec{U} | \mat{X}]$ is known, or more practically, consistent if $\E[\vec{U} | \mat{X}]$ is consistently estimated.

Since $\vec{U}$ is unmeasured, direct modeling of $\E[\vec U | \mat X]$ based on traditional estimation methods is {\it not} possible.
Hence, identifiability of this component and our ability to calculate the affine estimator require additional assumptions. In \cref{subsec:construction}, we provide a set of assumptions based on Gaussian Markov random field theory which pertain to the joint distribution of $(\vec U, \vec Z) | \mat{X}_{(-z)}$.
Based on these assumptions, we discuss an approach to calculating the affine estimator in the context of restricted maximum likelihood in \cref{subsec:affine_reml}.
Then, in \cref{subsec:identifiability} we show that these assumptions form a sufficient set for identification of the components of $\vec U$ on which the affine estimator is based, and therefore the estimation procedure is sound. The identifiability results illustrate that identification is achieved by exploiting the spatial correlation structure in the exposure and outcome model residuals which is driven by the unmeasured spatial variable.

\subsection{A sufficient set of assumptions}
\label{subsec:construction}

In this section, we present a set of assumptions which pertain to both the spatial and causal aspect of the affine estimator, and are summarized in \cref{tab:assumptions}. We proceed with these for ease of illustration, and because they seem plausible within our study setting (see \cref{sec:example-assumptions}), but note that different or weaker assumptions for identification of $\E[\mb U|\mat{X}]$ are likely possible  (see \cref{subsec:identifiability} and \cref{sec:discussion}).

\subsubsection{A Gaussian Markov random field construction of the joint distribution}

In viewing the model from a spatial perspective  and to better align to the spatial modeling literature, we assume that the marginal distributions of $\vec{U}$ and $\vec{\epsilon}$ are Gaussian with mean zero, independent of the measured covariates $\mat{X}_{(-z)}$, and that $(\vec{U}, \vec{Z}) | \mat{X}_{(-z)}$ is multivariate normal. We see the assumption that $\vec U$ is independent of $\mat{X}_{(-z)}$ as without loss of generality, since the same procedure could be alternatively followed for $\vec U - \mathbb{P}(\vec U | \mat{X}_{(-z)})$, where $\mathbb{P}(\vec U | \mat{X}_{(-z)})$ is the projection of $\vec U$ on the column space of $\mat{X}_{(-z)}$. This is also supported by results from simulated scenarios under which $U$ and $\mat{X}_{(-z)}$ are correlated (see \cref{sec:simulation}).
Further, we recognize that joint normality might be a strong assumption and we discuss an approach to relaxing it in \cref{sec:discussion}.
We make the following assumptions about the joint distribution of $(\vec U, \vec Z) | \mat{X}_{(-z)}$.
\begin{enumerate}
\item \textbf{Cross-Markov property:} $p(Z_i | \vec{Z}_{-i}, \vec{U}, \mat{X}_{(-z)}) = p(Z_i | \vec{Z}_{-i}, U_i, \mat{X}_{(-z)})$,
\item \textbf{Constant conditional correlation:} $\Cor(U_i, Z_i | \vec{U}_{-i}, \vec{Z}_{-i}, \mat{X}_{(-z)}) = \rho$.
\end{enumerate}
The first assumption states that, conditional on measured covariates and the values of $Z$ at all other locations, $Z_i$ depends on $\vec U$ only through its value at location $i$, $U_i$. Thus, it accommodates correlation between nearby treatments, but it does not allow $U_j$ to directly affect the value of $Z_i$ for $i\neq j$. The second assumption states that the conditional correlation between $U_i$ and $Z_i$ does not vary by location.
In the joint distribution of $(\vec U, \vec Z) | \mat{X}_{(-z)}$, these assumptions can be incorporated in the precision matrix (see Appendix~\ref{app:partial-corr} for derivations). Specifically, if
\begin{equation}
  \label{eq:model-joint}
  \begin{pmatrix}
    \vec{U} \\
    \vec{Z}
  \end{pmatrix}
  \sim \N\left[
    \begin{pmatrix}
      \vec{0} \\
      \mat{X}_{(-z)} \vec{\gamma}
    \end{pmatrix},
    \begin{pmatrix}
      \mat{G} & \mat{Q} \\
      \mat{Q}^\trans & \mat{H}
    \end{pmatrix}^{-1}
  \right],
\end{equation}
the cross-Markov assumption is equivalent to diagonal $\mat{Q}$, and along with the constant conditional correlation assumption they imply that
\begin{equation}
  \label{eq:q-spec}
  \begin{aligned}
    q_{ij} &= \left\{\begin{array}{lr}
                       -\rho \sqrt{g_{ii} h_{ii}}, & i = j, \\
                       0, & i \neq j.
                     \end{array}\right.
  \end{aligned}
\end{equation}

Given the above framework, the joint model for $\vec{U}$ and $\vec{Z}$ is completed by specifying $\mat{G}$ and $\mat{H}$, the precision matrices of $\vec{U} | (\vec{Z}, \mat{X}_{(-z)})$ and $\vec{Z} | (\vec{U}, \mat{X}_{(-z)})$ respectively, up to some \textit{unknown} parameters that will be estimated from the data.
For areal data like the ones in \cref{sec:example}, we adopt conditional autoregressive structures (CAR; \cite{besag1974spatial}) for $\mat{G}$ and $\mat{H}$, a common assumption in standard spatial analysis models.
Then, the precision matrices $\mat{G}$ and $\mat{H}$ are assumed to share the same neighborhood structure which is encoded in the matrices $\mat{D}$ and $\mat{W}$ of \cref{eq:multivariate_CAR},
but are allowed to differ by their precision and spatial dependence parameters $(\tau_U, \phi_U)$ and $(\tau_Z, \phi_Z)$. Based on \cref{eq:q-spec}, the assumed CAR structure leads to $\mat{Q} = -\rho\sqrt{\tau_U \tau_Z}\ \mat{D}$.

In the analysis of point-referenced data, the precision matrices of $\vec U | (\vec Z, \mat{X}_{(-z)})$ and $\vec Z | (\vec U, \mat{X}_{(-z)})$ can be specified based on a correlation function decaying in geographical distance. In either case, since $\mb U$ includes all unmeasured spatial variables $\covs^u$, the correct specification of its precision matrix $\mat G$ becomes harder for a larger number of unmeasured spatial covariates. We note again here that, since $\vec U$ is unmeasured, estimating the components of the joint distribution in \cref{eq:model-joint} that contribute to the affine estimator {\it cannot} be based on traditional modeling approaches, and instead is based on harvesting information from the spatial structure in exposure and outcome model residuals (as we see in \cref{subsec:affine_reml} and \cref{subsec:identifiability}).

\begin{table}[!t]
\centering \small
\caption{Set of assumptions based on which the causal exposure-response curve derivative is identifiable using observed data and can be estimated using the affine estimator.}
\begin{tabular}{p{6.5cm} p{9cm}}
\\
\hline \hline
Causal Assumptions    \\ \hline \\
Temporal Order & The exposure is temporally precedent to the outcome \\[5pt]
SUTVA    & No interference between units, no hidden levels of the treatment, $Y_i = Y_i(Z_i)$ \\[5pt]  
No unmeasured non-spatial confounding & $Z \independent Y(z) |\covs^m, U, z \in \col{Z}$ \\[5pt]
Positivity$^*$ & $P(Z = z | \covs^m, U) > 0 , z \in \col{Z}$, which implies that: \\[5pt]
Spatial scale restriction &  $Z$ has variation at a smaller spatial scale than that of $U$ \\ \\
Structural Assumptions \\ \hline \\
Outcome Additivity & The exposure and measured covariates do not interact with the unmeasured covariates. \\ \\
$E[\vec U | \vec Z]$ identification \\ \hline \\
Normality & $Z, U, \epsilon$ are jointly normal, conditional on $\mat{X}_{(-z)}$ \\[5pt]
Cross-Markov & $Z_i \!\perp\!\!\perp \vec U_{-i} | U_i, \vec Z_{- i}, \mat{X}_{(-z)} $, where\\
& $\vec U_{-i} = (U_1, U_2, \dots, U_{i -1}, U_{i + 1}, \dots, U_n)^T$, $\vec Z_{-i}$ defined similarly \\[5pt]
Conditional correlation & $\text{Cor}(U_i, Z_i | \vec U_{-i}, \vec Z_{-i}, \mat{X}_{(-z)})$ is constant \\[5pt]
Precision matrices & The precision matrices of $\vec{U} | \vec{Z}, \mat{X}_{(-z)}$ and $\vec{Z} | \vec{U}, \mat{X}_{(-z)}$ are of CAR form \\[5pt]
Mean Specification & $E[\vec Y | \mat X, \vec U]$ and $E[\vec Z | \mat X_{(-z)}, \vec U]$ are correctly specified.
\\ 
\hline\hline
\end{tabular}
\label{tab:assumptions}
\begin{flushleft} \footnotesize
$^*$For continuous exposures, positivity can be defined in terms of the probability {\it density} function of $Z$ conditional on measured and unmeasured variables.
\end{flushleft}
\end{table}

\subsubsection{Spatial scale restriction for the unmeasured spatial confounder}
\label{subsec:scale_restriction}

In order to draw causal conclusions using our approach, the no-unmeasured and positivity assumptions still need to hold, conditional on the measured covariates $\covs^m$, and the unmeasured $U$. The assumption of positivity implies that estimation of the causal effect of $Z$ on $Y$ in the presence of $U$ is only possible if there is variability in $Z$ within levels of $U$. If the spatial scale of $U$ is smaller than that of $Z$, the positivity assumption will be violated, since, loosely speaking, there may be ``strata'' of $U$ within which only one value of $Z$ is possible. Therefore, from a causal perspective, we assume that the spatial scales of the exposure and spatial confounder do not violate positivity of the treatment assignment conditional on the unmeasured spatial confounder and the measured covariates.  In \cref{subsec:mediation}, we also discuss how the spatial scale restriction is also useful in settings where spatial variables mediate the effect of interest.

The spatial scale restriction has been studied from a spatial perspective. \cite{paciorek2010importance} shows that the bias and variance of spatial model estimators depend on the relative spatial scales of the exposure and the residual including the confounder, $\epsilon + U$. \cite{paciorek2010importance} recommends only fitting spatial models when there is exposure variation on a spatial scale smaller than that of the unmeasured confounder, essentially ensuring positivity. From a spatial perspective, the spatial scale restriction ensures that we do not mistakenly attribute all spatial variability of the outcome residuals to the unmeasured spatial confounder when it is truly due to the exposure.

The spatial scale restriction can be enforced through the precision matrices $\mat{G}, \mat{H}$ in \cref{eq:model-joint}.
For geostatistical data, the spatial scale of dependence is often an explicit modeling parameter, as in \cite{paciorek2010importance}.
For a conditional autoregressive model of areal data, the autocorrelation parameters $\phi_Z$ and $\phi_U$ do not have strict interpretations as spatial scale parameters, though the restriction $\phi_Z < \phi_U$ plays a similar role in identifying variance parameters.

\subsection{The affine estimator within a restricted likelihood framework}
\label{subsec:affine_reml}

In this section we describe estimation within a restricted likelihood framework. We do so because it allows for straightforward illustration of how model components correspond to components in the bias results of \cref{sec:estimators} and the affine estimator in \cref{eq:est-joint}. Further, it allows us to easily discuss identifiability of the model parameters in \cref{subsec:identifiability}.
In \cref{subsec:affine_bayesian}, we describe a fully-Bayesian approach to estimation which is applicable for linear and non-linear models, and which we follow for the remainder of this paper.

\subsubsection{Linear effect estimator}
\label{subsec:linear_estimation}

We start by assuming the linear structure in \cref{eq:dgm}. Using the conditional distribution $\mb U | \mat{X}$ acquired from \cref{eq:model-joint}, the joint model for the observed data (integrating $\mb U | \mat{X}$ out) can be factored as
\begin{equation}
  \label{eq:observation-model-factored}
  \begin{aligned}
    \vec{Y} | \mat{X} &\sim \N[\mat{X} \vec{\beta} - \mat{G}^{-1} \mat{Q} (\vec{Z} - \mat{X}_{(-z)}\vec{\gamma}), \mat{G}^{-1} + \mat{R}^{-1}], \\
    \vec{Z}|\mat{X}_{(-z)} &\sim \N[\mat{X}_{(-z)} \vec{\gamma}, (\mat{H} - \mat{Q}^\trans \mat{G}^{-1} \mat{Q})^{-1}].
  \end{aligned}
\end{equation}
where $\mat R^{-1} = \mathrm{Cov}(\vec \epsilon)$ (see Appendix~\ref{app_sec:marginal_variances} for the derivation).
From \cref{eq:observation-model-factored}, we see that the likelihood depends on $\vec U$ through the components of the precision matrix in \cref{eq:model-joint}. Note that, even though our focus is in estimating parameters of the outcome model ($\vec \beta$), an exposure model is also adopted to provide information on the spatial structure of $\vec{U}$. (This is related to many settings in causal inference where incorporating information from the exposure model improves estimation of causal effects \citep[e.g.][]{wilson2014confounder, belloni2014inference, antonelli2019high}.)

Following a common approach to estimation for mixed models, variance parameters are estimated based on the restricted likelihood derived from \cref{eq:observation-model-factored}, and the estimates are used to calculate the bias-adjusted affine estimator $\widebar{\vec{\beta}}$ in \cref{eq:est-joint}.
Defining
\begin{equation}
\label{eq:rl-components}
    \begin{aligned}
        \mat{M}
        &=
        \begin{pmatrix}
            \mat{G}^{-1} + \mat{R}^{-1} & \mat{0} \\
            \mat{0} & (\mat{H} - \mat{Q}^\trans \mat{G}^{-1} \mat{Q})^{-1}
        \end{pmatrix}, \\
        \mat{C}
        &=
        \begin{pmatrix}
            \mat{X} & \mat{G}^{-1} \mat{Q} \mat{X}_{(-z)} \\
            \mat{0} & \mat{X}_{(-z)}
        \end{pmatrix}, \\
        \vec{\nu}
        &=
        \begin{pmatrix}
            \vec{Y} + \mat{G}^{-1} \mat{Q} \vec{Z} \\
            \vec{Z}
        \end{pmatrix}, \text{ and } \\
        \vec{\theta}
        &=
        \begin{pmatrix}
            \vec{\beta} \\
            \vec{\gamma}
        \end{pmatrix},
    \end{aligned}
\end{equation}
we can write the joint distribution of $(\vec{Y}, \vec{Z})$ as
\begin{equation}
    f(\vec{Y}, \vec{Z} | \vec{\beta}, \vec{\gamma})
    \propto|\mat{M}|^{-\frac{1}{2}} \exp\left[
        -\frac{1}{2} (\vec{\nu} - \mat{C}\vec{\theta})^\trans \mat{M}^{-1} (\vec{\nu} - \mat{C}\vec{\theta})
    \right],
\end{equation}
and the restricted likelihood as
\begin{equation}
\label{eq:res-lik}
    \begin{aligned}
    RL
    &\propto \left[|\mat{M}| \cdot |\mat{C}^\trans \mat{M}^{-1} \mat{C}| \right]^{-1/2} 
     \exp\left[-\frac{1}{2}
        \begin{array}{r}
          \vec{\nu}^\trans
      \left( \mat{M}^{-1} - \mat{M}^{-1} \mat{C} (\mat{C}^\trans \mat{M}^{-1} \mat{C})^{-1} \mat{C}^\trans \mat{M}^{-1} \right)
          \vec{\nu}
        \end{array}
    \right].
    \end{aligned}
\end{equation}
If $\widehat{\mat{M}}$, $\widehat{\mat{C}}$, and $\widehat{\vec{\nu}}$ are maximizers of the restricted likelihood in \eqref{eq:res-lik}, we calculate $(\widebar{\vec{\beta}}, \widebar{\vec{\gamma}}) = (\widehat{\mat{C}}
^\trans \widehat{\mat{M}}^{-1} \widehat{\mat{C}})^{-1} \widehat{\mat{C}}^\trans \widehat{\mat{M}}^{-1} \widehat{\vec{\nu}}$. (Readers interested in the REML approach can find additional information in Appendix~\ref{app_sec:restricted_likelihood}.)

The restricted likelihood formulation allows us to make illuminating connections between our approach, the mixed effects models often used in spatial statistics, and the bias results of existing approaches in \cref{sec:estimators}. If $\rho = 0$, the matrix $\mat Q$ is zero, the model in \cref{eq:model-joint} reduces to the case where $\vec U, \vec Z$ are independent, and the restricted likelihood estimation method would lead to the estimator $\widetilde{\vec \beta}$. A non-zero correlation $\rho$ between $\vec U$ and $\vec Z$ leads to the inclusion of the $-\mat G^{-1}\mat Q$ component in the coefficient of $\vec Z$, corresponding to the bias correction term $\E[\vec U | \mat{X}] = -\mat G^{-1}\mat Q (\vec Z - \mat{X}_{(-z)}\vec{\gamma})$.

\subsubsection{Semi-parametric effect estimator}
\label{sec:semiparametric}

To better accommodate continuous exposures, we can flexibly model the exposure-response relationship using penalized regression splines. Penalized regression splines may be represented as linear mixed models \citep{ruppert2003semiparametric}, allowing for the linear effect assumption in \cref{eq:dgm} to be relaxed to
\begin{equation}
  \label{eq:dgm-nonlinear}
  \vec{Y} = \vec{1} \beta_0 + \vec{f}(\vec{Z}) + \mat{X}_{(-z)} \vec{\beta}_{(-z)} + \vec{U} + \vec{\epsilon},
\end{equation}
where $\vec{f}(\vec{Z}) = (f(Z_1), \ldots, f(Z_n))^\trans$ and $f$ being a smooth function of $Z$.
Our chosen radial basis penalized spline model for $f$ may then be written as
\begin{equation}
  \label{eq:spline}
  \widebar{f}(z) = \sum_{a = 1}^{A} \beta_a z^a + \sum_{k=1}^{K} l_k |z - \xi_k|^A,
\end{equation}
where $A$ is the degree of the spline and the $\xi_k$ are pre-specified knots.
Letting
\begin{equation}
  \mat{X} =
  \begin{pmatrix}
    1 & z_1^1 & \cdots & z_1^A \\
    \vdots & \vdots & \ddots & \vdots \\
    1 & z_n^1 & \cdots & z_n^A
  \end{pmatrix},
  \quad 
  \mat{L} =
  \begin{pmatrix}
    |z_1 - \xi_1|^A & \cdots & |z_1 - \xi_K|^A \\
    \vdots & \ddots & \vdots \\
    |z_n - \xi_1|^A & \cdots & |z_n - \xi_K|^A
  \end{pmatrix},
\end{equation}
and $\mat{V} = \psi^{-1} \mat{L} \mat{L}^\trans + \mat{G}^{-1} + \mat{R}^{-1}$, where $\psi > 0$ is a roughness penalty, the restricted likelihood is as in \cref{eq:res-lik} with updated components $\mat C$ and $\mat M$.
Letting $\mat{T} = (\mat{X} \; \mat{L})$ and $\mat{A}$ be a diagonal matrix with the first $A + 1$ elements equal to $0$ and the rest equal to $1$ (corresponding to penalization of the $\vec \beta$ and $\vec l$, respectively), the estimate of $\vec{\theta} = (\beta_0, \beta_1, \ldots, \beta_A, l_1, l_2, \ldots, l_K)$ is $\widebar{\vec{\theta}} = \left(\mat{T}^\trans \widebar{\mat{V}}^{-1} \mat{T} + \psi \mat{A} \right)^{-1} \mat{T}^\trans \widebar{\mat{V}}^{-1} \left(\vec{Y} - \widebar{\mat{B}} \vec{Z} \right)$.

\subsection{Identifiability of parameters}
\label{subsec:identifiability}

It is evident from the form of the affine estimator in \cref{eq:est-joint} and the restricted likelihood in \cref{eq:res-lik} that calculating the estimator depends on being able to estimate components of the relationship of the unmeasured confounder with the exposure and outcome of interest. Therefore, it is natural to wonder whether these components can be identified, and if so, which assumptions are key in driving identifiability and which can be relaxed.
We provide two examples: one in which identifiability can be proved analytically and one in which model components are not identifiable.
The details of how identifiability is achieved or lost in these examples illuminate the key assumptions and their roles.

\subsubsection{Identifiability of model components for a ring graph}
\label{subsec:identifiability_ring}

The first example is one in which identifiability can be analytically proved. We consider the setting without measured variables since including them complicates the notation without providing any additional insight, and results extend trivially. We also assume that the spatial dependence can be represented in a ring graph of $n$ locations depicted in \cref{fig:ring}, with CAR specifications for $\vec{U} | \vec{Z}$ and $\vec{Z} | \vec{U}$.
This spatial structure and model yields the precision matrix
\begin{equation}
    \mat{G}_n = \tau_U
  \begin{pmatrix}
    2 & -\phi_U & & & -\phi_U \\
    -\phi_U & 2  & -\phi_U & & \\
    & \ddots & \ddots & \ddots & \\
    & & -\phi_U & 2 & -\phi_U & \\
    -\phi_U & & & -\phi_U & 2
  \end{pmatrix}
\end{equation}
for $\vec{U} | \vec{Z}$ and similar for $\mat{H}_n$ for $\vec{Z} | \vec{U}$.
For simplicity, we assume that $\E[\vec{Z}] = \vec{0}$ (marginally over $\vec{U}$) and that $\E[\vec{Y} | \vec{Z}, \vec{U}] = \beta_Z \vec{Z} + \vec{U}$.
We present two results: the first discussing parameter identifiability based on $\vec{Z}$ alone (\cref{thm:identification-z}), and the second based on $(\vec{Y}, \vec{Z})$ (\cref{thm:identification-y}). Importantly, these identifiability results allow for $\vec U$ to be completely unmeasured.

\begin{theorem}
\label{thm:identification-z}
In the scenario defined in this section, it can be determined whether or not $\rho \phi_U = 0$ by observing $\vec{Z}$.
Further, if $\rho \phi_U \neq 0$, then $(\tau_Z, \phi_Z, \phi_U, |\rho|)$ is also identifiable by observing $\vec{Z}$.
\end{theorem}

\begin{proof}
The proof relies on a few matrix lemmas which are stated and given proof outlines in Appendix~\ref{app:identifiability}.
We have that $\Prec[\vec{Z}] = \tau_Z \left[\mat{H}_n - 4\rho \mat{G}_n^{-1}\right]$ (marginally over $\vec{U}$), and
\begin{equation}
    \begin{aligned}
    \lim_{n\to\infty} \Prec[\vec{Z}]_{ij}
      &= \begin{cases}
      \tau_Z \left[2 -  2\rho ^ 2 \frac{1}{\sqrt{1-\phi_U^2}} \right], & i = j, \\[10pt]
      \tau_Z \left[-\phi_Z -  2\rho ^ 2 \frac{\phi_U}{\sqrt{1-\phi_U^2}\left(1 + \sqrt{1-\phi_U^2}\right)} \right], & |i - j| = 1, \\[12pt]
      \tau_Z \left[ 0 -  2\rho ^ 2 \frac{\phi_U^{|i-j|}}{\sqrt{1-\phi_U^2}\left(1 + \sqrt{1-\phi_U^2}\right)^{|i-j|}} \right], & |i - j| > 1.
      \end{cases}
    \end{aligned}
    \label{eq:limit_prec}
\end{equation}
First, note that for any lag $l$, the number of pairs of locations with $|i - j| = l$ grows linearly with $n$.
It can be determined whether or not $\rho \phi_U = 0$ because $\lim_{n\to\infty} \Prec[\vec{Z}]_{ij} = 0$ for all $(i, j)$ such that $|i - j| > 1$ if and only if $\rho \phi_U = 0$.
If $\rho \phi_U \neq 0$, then for $|i - j| > 1$ and $|i^\prime - j^\prime| = |i - j| + 1$ the ratio $\lim_{n\to\infty} \Prec[\vec{Z}]_{ij} / \Prec[\vec{Z}]_{i^\prime j^\prime}$ depends only on $\phi_U$ and is bijective, thus $\phi_U$ is identified.
With $\phi_U$ identified, the three cases in \cref{eq:limit_prec} form a system of equations solvable for $(\tau_Z, \phi_Z, |\rho|)$.
\end{proof}

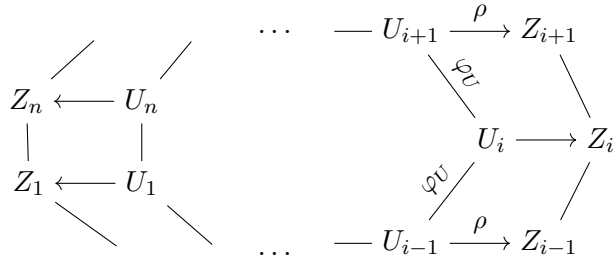
\begin{figure}[!b]
	\centering
	\begin{tikzpicture}[paths/.style={-}]
	% nodes %
	
	\node[] (U1) {$U_1$};
	\node[above= 0.5 of U1] (Un) {$U_n$};
	\draw[-] (U1) -- (Un);
	\node[left = 0.8 of U1] (Z1) {$Z_1$};
	\draw[->] (U1) -- (Z1);
	\node[left = 0.8 of Un] (Zn) {$Z_n$};
	\draw[->] (Un) -- (Zn);
	\draw[-] (Z1) -- (Zn);
	
	\node[below=0.5 of U1] (useless1) {};
	\node[right = 0.8 of useless1] (useless12) {};
	\node[right=0.2 of useless12] (dots1) {$\dots$};
	\draw[-] (U1) -- (useless12);
	\node[left = 1 of useless12] (useless13) {};
	\draw[-] (Z1) -- (useless13);
	
	\node[above = 0.5 of Un] (uselessn) {};
	\node[right = 0.5 of uselessn] (uselessn2) {};
	\node[right= 0.5 of uselessn2] (dotsn) {$\dots$};
	\draw[-] (Un) -- (uselessn2);
	\node[left = 1 of uselessn2] (uselessn3) {};
	\draw[-] (Zn) -- (uselessn3);

	\node[right = 0.1 of dotsn] (uselessn3) {};
	\node[right = 0.5 of uselessn3] (Uip) {$U_{i + 1}$};
	\draw[-] (uselessn3) -- (Uip);
	\node[right = 0.8 of Uip] (Zip) {$Z_{i + 1}$};
	\draw[->] (Uip) -- (Zip) node [midway,above] (TextNode5) {\small $\rho$};
	
	\node[below=1 of Zip] (uselessui) {};
	\node[left=0.3 of uselessui] (Ui) {$U_i$};
    \draw (Uip) -- (Ui) node [midway, above, sloped] (TextNode) {\small $\phi_U$};

	\node[right = 0.8 of Ui] (Zi) {$Z_i$};
	\draw[->] (Ui) -- (Zi);
	
	\node[below = 2.2 of Zip] (Zim) {$Z_{i - 1}$};
	\node[left = 0.8 of Zim] (Uim) {$U_{i - 1}$};
    \draw (Uim) -- (Ui) node [midway, above, sloped] (TextNode2) {\small $\phi_U$};
	\draw[->] (Uim) -- (Zim) node [midway,above] (TextNode4) {\small $\rho$};
	\node[left = 0.5 of Uim] (uselessim) {};
	\draw[-] (uselessim) -- (Uim);
	
	\draw[-] (Zi) -- (Zim);
	\draw[-] (Zi) -- (Zip);
	
	\end{tikzpicture}
\caption{Graph representation of ring with the cross-Markov property and dependence parameters.}
\label{fig:ring}
\end{figure}

Based on \cref{thm:identification-z}, if $\rho\phi_U = 0$ some parameters are not identifiable by only observing $\vec Z$. That is because, when $\rho = 0$, $Z$ and $U$ are uncorrelated, hence $Z$ cannot provide any information on $U$, and when $\phi_U = 0$, the unmeasured variable is not spatial and the variability in $Z$ cannot be decomposed accordingly (a situation we examine closely in \cref{subsec:non_spatial}).
In contrast, when $\rho\phi_U \neq 0$, a number of spatial parameters are identified based solely on the vector of treatments $\vec Z$. For intuition about why the unobserved variable's spatial dependence parameter $\phi_U$ (but not its precision $\tau_U$) can be identified by examining the behavior of $\Prec[\vec{Z}]$ away from the tri-diagonal, recall that off of the tri-diagonal, $i$ and $j$ are not neighbors, and that the precision in \cref{eq:limit_prec} tells us about the strength of dependence between $Z_i$ and $Z_j$ given the value of $\vec{Z}$ at all other locations.
Our cross-Markov assumption states that, {\it conditionally on $\vec{U}$}, the values of $\vec{Z}$ at non-neighboring locations are independent given the values of $\vec{Z}$ at other locations.
If $Z_i$ and $Z_j$ are not independent conditional only on $\vec{Z}$ (and not on $\vec{U}$), this dependence has to arise through paths in \cref{fig:ring} that pass through $\vec{U}$.
The strength of dependence is a function of $\rho$, the strength of connection between $U$ and $Z$ at a given location (which is independent of the distance between $i$ and $j$), and $\phi_U$ which determines how quickly the dependence of $U_i, U_j$ attenuates with distance $|i - j|$. These dependencies are graphically represented in \cref{fig:ring} where the dependence between $Z_{i - 1}$ and $Z_{i + 1}$, conditional on $\vec{Z}$ at other locations but marginally over $\vec{U}$, arises from the paths $Z_{i - 1}-U_{i -1} - U_i - U_{i + 1} - Z_{i + 1}$ and $Z_{i-1} - U_{i-1} - U_{i-2} - \cdots - U_{1} - U_{n} - \cdots U_{i + 1} - Z_{i+1}$, and the dependency due to the latter path diminishes as $n$ becomes large.
We can identify $\phi_U$ off of the tri-diagonal by examining this attenuation, and $\rho$ is separable from $\tau_Z$ only when examining the tri-diagonal as well.

\begin{theorem}
\label{thm:identification-y}
The parameter $\phi_U$ is identifiable by observing $(\vec{Y}, \vec{Z})$
Further, if $\phi_U \neq 0$ then the full parameter $(\beta_Z, \tau_Z, \phi_Z, \tau_U, \phi_U, \rho, \tau_{\epsilon})$ is identifiable by observing $(\vec{Y}, \vec{Z})$.
\end{theorem}

\begin{proof}
Note that $\E[\vec{Y} | \vec{Z}]$ is identified irrespective of $\Var[\vec{Y} | \vec{Z}]$, and that by Theorem~\ref{thm:identification-z} we can identify whether $\rho \phi_U = 0$ by observing $\vec{Z}$.

We start by showing that $\phi_U$ is identifiable by observing $(\vec Y, \vec Z)$. Since $\Var[\vec{Y} | \vec{Z}] = \mat{G}_n^{-1} + \tau_{\epsilon}^{-1} \mat{I}_n$, by noting the similarity between the expressions for $\Var[\vec{Y} | \vec{Z}]$ and $\Prec[\vec{Z}]$ in \eqref{eq:limit_prec}, we have
\begin{equation}
\begin{aligned}
    \lim_{n\to\infty} \Var[\vec{Y}|\vec{Z}]_{ij} = 
    \tau_U^{-1} \frac{1}{2\sqrt{1-\phi_U^2}} \left(\frac{\phi_U}{1+\sqrt{1-\phi_U^2}}\right)^{|i-j|}, && i \neq j.
    \end{aligned}
\label{eq:var_y_giv_z}
\end{equation}
Therefore, $\lim_{n\to\infty} \Var[\vec{Y} | \vec{Z}]_{ij} = 0$ for $i \neq j$ if and only if $\phi_U = 0$. Since we can identify $\Var[\vec{Y} | \vec{Z}]$, $\phi_U$ is identifiable by $\lim_{n\to\infty} \Var[\vec{Y}|\vec{Z}]_{ij} / \Var[\vec{Y}|\vec{Z}]_{i^\prime j^\prime}$ for $|i - j| > 0$ and $|i^\prime - j^\prime| = |i - j| + 1$.

Next, assume that $\phi_U \neq 0$, and we show that the remaining parameters are identifiable by observing $(\vec{Y}, \vec{Z})$. We first note that since
$\E[\vec{Y} | \vec{Z}] = \left(\beta_Z - \rho\sqrt{\tau_Z / \tau_U}
\right) \vec{Z}$, the {\it combined} coefficient of $\vec Z$, $\beta_Z - \rho\sqrt{\tau_Z / \tau_U}$ is identifiable.
Since we can identify whether $\rho\phi_U \neq 0$ from observing $\vec Z$ (\cref{thm:identification-z}), and since we have here that $\phi_U \neq 0$, we can identify whether $\rho = 0$, which allows us to consider the cases where $\rho = 0$ and $\rho \neq 0$ separately.

If $\rho = 0$, the combined coefficient of $\vec Z$ is equal to $\beta_Z$, and $\beta_Z$ is identified (which is expected since for $\rho = 0$ there is no confounding by $\vec{U}$).
Additionally, we can return to \eqref{eq:limit_prec} to identify $(\tau_Z, \phi_Z)$ from the first two cases.
Finally, with $\phi_U \neq 0$ identified, $\tau_U$ can be identified from the off-diagonal elements of $\Var[\vec{Y} | \vec{Z}]$ in \cref{eq:var_y_giv_z}.

If $\rho \neq 0$, we can identify $(\tau_Z, \phi_Z, \phi_U, |\rho|)$ by Theorem~\ref{thm:identification-z}. Recall that $\E[\vec{Y} | \vec{Z}] = \big( \beta_Z - \mat{G}_n^{-1} \mat{Q}_n \big) \vec{Z}$ and $-\mat{G}_n^{-1}\mat{Q}_n = 2\rho\sqrt{\frac{\tau_Z}{\tau_U}}\left(\tau_U\mat{G}_n^{-1}\right)$, where $\tau_U \mat{G}_n^{-1}$ does not depend on $\tau_U$.
Since $\phi_U \neq 0$, $\mat{G}_n$ is not a scalar matrix, and the terms $\vec{Z}$ and $\tau_U \mat{G}_n^{-1}\vec{Z}$ are known and not collinear.
Thus we can separately identify their coefficients $\beta_Z$ and $\rho\sqrt{\tau_Z / \tau_U}$.
Since $\tau_U$ and $\tau_Z$ are both positive, and $(\tau_Z, |\rho|)$ has been previously identified, we can identify $\sign(\rho)$, and $\tau_U$.

In both cases ($\rho = 0$ or $\rho \neq 0$), $\tau_{\epsilon}$ can be identified from the diagonal elements of $\Var[\vec{Y} | \vec{Z}]$.

\end{proof}

% The above treatment precludes the case where $\tau_U^{-1} = 0$, but allows it to be arbitrarily small, yielding a close-to-constant $U$ and effectively no confounding of the effect of $\vec{Z}$ on $\vec{Y}$.% but allowing non-Markov behavior in $\vec{Z}$.

Theorems  \ref{thm:identification-z} and \ref{thm:identification-y} combined establish that when spatial dependencies can be represented using a ring, and the unmeasured variable is spatial ($\phi_U \neq 0$) and is truly a confounder ($\rho \neq 0$), the effect of the exposure $\vec{Z}$, $\beta_Z$, is identifiable.
%nFurther, the proof of \cref{thm:identification-y} illustrates that identifiability can be easily extended to the situation where the elements in $\vec Y$ are not independent conditional on $\vec U$, and $\vec Z$, and are themselves spatially structured.\noteG{Isn't this true? For example, CAR prior on Var(Y given Z) would mean additional variance elements to be identified, but can be done.}

The above results and the details of the proofs indicate that the cross-Markov property is the critical assumption in identifying $\beta_Z$.
This is not to say that we can remove the more parametric or the constant conditional correlation assumptions, but rather that these assumptions are most likely not the only or weakest ones allowing for identification.
In turn, the cross-Markov assumption can be viewed as a relaxation of the usual assumption of no unobserved confounders, expressing that confounding by unobserved spatial variables is local.
To see this, the assumption may be rewritten as $Z_i \independent \vec{U}_{-i} | (\vec{Z}_{-i}, U_i, \mat{X}_{(-z)})$, expressing that conditional on measured covariates, the value of $U$ at location $i$, and the value of $Z$ everywhere else, $U_k$ with $k \neq i$ is not a predictor of $Z_i$ and therefore does not confound the relationship between $Z_i$ and $Y_i$.
% The consequence of this assumption is that we can determine whether or not the bias-inducing term $-\mat{G}_n^{-1}\mat{Q}_n\vec{Z}$ is collinear with $\vec{Z}$, and if it is not collinear, we can identify $\beta_Z$.

\iffalse
In the proof of Theorem~\ref{thm:identification-y} we do not necessarily need to identify $\rho$, $\tau_z$, and $\tau_U$ separately.
Rather, the important point is to have identified from observing $\vec{Z}$ a transformation of $\vec{Z}$ that is known up to a scaling factor and not collinear with $\vec{Z}$.
This transformation is made non-collinear with $\vec{Z}$ by the spatial structure of $\vec{U}$, which we learn about chiefly by examining how the dependence of $Z$ at non-neighboring points, conditional on $\vec{Z}$ elsewhere (which would be zero in the absence of a spatial confounder), attenuates with distance.
This suggests that given a sufficiently large location space, so that there are enough pairs of non-neighboring observations (corresponding to zeroes in $\Prec[\vec{Z}|\vec{U}]$) with varying lags, we should be able to identify a representation of spatial conditional dependence attenuation function so long as its parameterization does not grow too quickly with $n$, especially if some sort of stationarity assumption is available.
This in turn suggests identification of $\vec{\beta}$.
\fi

Although we have proved these asymptotic results for a ring of locations, the key requirement on the structure of a graph with large connected components is that there are enough pairs of locations at varying lags.
This requirement does not seem problematic for, e.g., counties in the United States, as in our study in \cref{sec:example}.

\subsubsection{A non-sufficient set of assumptions: unmeasured non-spatial confounders}
\label{subsec:non_spatial}

Within framework \cref{eq:dgm-general}, any set of assumptions that suffice for identification of $\E[\vec U | \mat{X}]$ would also allow for identification of the causal parameter $\mu'(z)$. 
In \cref{subsec:identifiability_ring} we showed that for adjacency structures described as a ring of growing size, identifiability is achieved when the unmeasured confounder is spatial $(\phi_U \neq 0)$. Here, we establish that identifiability is lost when the unmeasured confounder is not spatial $(\phi_U = 0)$ under any adjacency structure.

Assume that $U_i$ are independent and identically distributed random variables (exhibiting no spatial structure) with $\mat{G} = \tau_U \mat{I}$. Further, assume that $\mb Z$ and $\mb \epsilon$ are not spatially structured with precision matrices $\mat{H} = \tau_Z \mat{I}$ and $\mat{R} = \tau_\epsilon \mat{I}$. Then, the parameter vector $(\tau_U, \tau_Z, \tau_\epsilon, \rho)$ in the restricted likelihood is reduced to $(\sigma^2, \phi) = (\tau_U^{-1} + \tau_\epsilon^{-1}, \tau_Z (1-\rho^2))$, and the parameters in $(\tau_U, \tau_\epsilon)$ and in $(\tau_Z, \rho)$ are not separately identifiable (the mathematical derivations are included in Appendix~\ref{app:non-spatial}).
Thus, when there is no spatial structure, $\widebar{\vec{\beta}} = (\mat{X}^\trans \mat{X})^{-1} \mat{X}^\trans (\vec{Y} - \rho \sqrt{\frac{\tau_Z}{\tau_U}}\vec{Z})$ is not identifiable either.

This result is intuitively obvious. If $\vec U$ is not spatially structured, there is no information in the observed data to differentiate outcome model residuals' variability due to $U$ from that due to $\epsilon$, and similarly nothing to differentiate intrinsic variability in $Z$ from variability due to $U$.
In such case, $\E[\vec{U} | \mat{X}]$ is not identifiable based on observed data indicating that adjustment for $U$ is not possible if the unmeasured confounders do not exhibit spatial structure. This is in line with recent work showing that latent variable approaches cannot be used to acquire identifiability of causal parameters without additional assumptions \citep{damour2019multiple,ogburn2019comment}.

% In the case where some of the unmeasured confounders in $\covs^u$ are spatial and others are not, identifiability (or lack thereof) of $\E[\vec{U} | \mat{X}]$ based on the observed data should be derived directly from the assumptions on the joint distribution and from the form of the restricted likelihood.

\subsection{Spatially correlated mediating variables}
\label{subsec:mediation}

From model \cref{eq:observation-model-factored} and the form of the restricted likelihood in \cref{eq:res-lik}, it is evident that information about the elements $(\G, \Q)$ in the bias correction term $-\G^{-1}\Q(\vec Z - \mat{X}_{(-z)}\vec{\gamma})$ is found in the spatial variability of exposure and outcome models' residuals.
However, the estimation procedure cannot differentiate between spatial structure arising from spatial confounders (temporally precedent of $Z$) or variables found on the causal pathway between $Z$ and $Y$ (mediators). %In fact, unmeasured spatial confounders could lead to similar spatial correlation structure in the exposure and outcome model residuals as could be observed in the presence of a spatial mediator.
If $\vec{Z}$ and $\vec{Y}$ are measured within a small time window, it may be reasonable to assume that there are no spatial covariates mediating the effect of $Z$ on $Y$, and for that reason, our estimates correspond to estimates of $\beta_z = \mu'(z)$.

On the other hand, in the presence of spatial intermediate variables, estimates using the affine estimator might more closely resemble the direct effect of $Z$ on $Y$, not due to changes to the spatial mediators \citep{Baron1986}. In this setting, the spatial scale restriction provides some protection against adjustment for spatial mediators since variation in spatial scales smaller than that of the exposure is not adjusted for. %, including spatial variation in mediating variables.
Therefore, the spatial scale restriction allows for unmeasured spatial confounder bias mitigation while protecting us from adjusting for variables on the causal pathway between $Z$ and $Y$.

\section{The affine estimator in non-linear settings: A Bayesian implementation}
\label{subsec:affine_bayesian}

The REML framework has allowed us to analytically investigate bias and identifiability, by allowing us to integrate out the distribution of the unmeasured confounder from the observed data likelihood. However, such analytical approach is less applicable to non-continuous outcomes and non-linear models. Here, we extend the affine estimator to non-linear models within the Bayesian framework.

\subsection{Estimation of causal parameters with non-linear models}
\label{subsec:other_outcomes}

When $\mu(z) = \E[Y(z)]$ is the estimand of interest, and the outcome model specifies $\E(Y | Z, \covs^m, U)$,
$\mu(z)$ can be written as
\(\displaystyle \E_{\covs^m, U} \left[ E(Y | Z = z, \covs^m, U) \right] \) under assumptions. This is often referred to as the g-formula, or g-computation \citep{robins1986new}.
If the outcome model is linear without exposure-covariate interactions, the regression coefficient for the exposure can be directly interpreted as an estimate of the causal quantity $\mu'(z)$. However, the exposure's coefficient cannot be directly causally interpreted in the case of non-linear models, even when the model is correctly specified.  
For example, in the context of logistic regression with binary outcomes, the coefficient of the exposure $\beta_z$ is not equal to $\mu'(z)$, and a linear specification of the exposure-response relationship in the linear predictor of the logistic regression does {\it not} imply a linear $\mu(z)$.
Therefore, in non-linear outcome models, we need to proceed with care when translating estimated coefficients to estimates of causal quantities, and an integration step (over the distribution of confounders in the target population) needs to be employed in order to acquire estimates of $\mu(z)$ from a non-linear model fit.

Poisson models for count outcomes like the one in our study are an exception. In such models, the parameter $\beta_z$ (or $\exp(\beta_z)$ which is often used in Poisson models) can be interpretable as causal, but for an estimand that is slightly different from $\mu'(z)$.
Let $P_i$ be the population at risk at location $i$, and $Y_i(z)$ be the potential outcome at location $i$ if the exposure was set to $z$.
Consider the PAERC defined in terms of the {\it standardized} outcome rate as $\E[P^{-1} Y(z)]$. Under a structural model similar to \cref{eq:dgm-general} for linear $\eta$, assume that $[Y(z) | \vec{X}_{(-z)}, U] \sim \text{Poisson} \Big(P \exp\{ \beta_z z + \beta_{(-z)}^\trans \bm X_{(-z)} + U \} \Big)$. Then, the PAERC can be written as
\begin{equation}
\E_{\vec{X}_{(-z)}, U} \left[ \exp\{ \beta_z z + \beta_{(-z)}^\trans \bm X_{(-z)} + U \} \right] =
\exp\{ \beta_z z \} \E_{\vec{X}_{(-z)}, U} \left[ \exp\{\beta_{(-z)}^\trans \bm X_{(-z)} + U \} \right].
\label{eq:poisson_estimand}
\end{equation}
From \cref{eq:poisson_estimand}, the coefficient $\beta_z$ can be interpreted as the log relative standardized rate for a one-unit exposure change, $\log \big\{ E[P^{-1} Y(z + 1)] \ / \ E[P^{-1} Y(z)] \big\}$. Equivalently, $\beta_z$ can be interpreted as the instantaneous effect of the exposure in the relative scale as $\partial \{ \log \E[P^{-1} Y(z)] \} / \partial z$.
Note that this is substantially different from the standard interpretation of estimated coefficients in Poisson models, and $\beta_z$ {\it cannot} be used as an estimate of $\log \big\{ E[ Y(z + 1) / Y(z)] \big\}$.

\subsection{Bayesian implementation of the affine estimator}

The presentation above indicates that, for non-linear models where non-collapsibility is an issue and estimated coefficients do not always estimate causal quantities, estimation might require an explicit model for $Y$ conditional on measured covariates {\it and} the unmeasured component.
This is straightforwardly achieved for the affine estimator within a Bayesian implementation, for which $\vec U$ is viewed as a missing variable that is iteratively imputed through a Gibbs sampler.
Therefore, placing the affine estimator within the Bayesian paradigm does not require marginalization over $\vec{U}$, which allows for estimation outside the realm of linear regression.
This is exploited in \cref{sec:simulation} and \cref{sec:example} where we consider a count outcome and a Poisson model with the log link and linear predictor $\eta_i = o_i + \vec{x}_i^\trans \vec{\beta} + u_i$, where $o_i$ is the offset in the usual sense, and sampling from the conditional posterior distribution of $\vec \beta$ can be performed without modifying standard algorithms for Poisson regression. 

Apart from its generalizability to non-linear models, the Bayesian approach has a number of additional benefits over the REML approach, including computational gains. The Bayesian implementation is computationally more efficient thanks to the conditional nature of Gibbs sampling which allows us to take advantage of sparsity in $\mat{P} = \Prec[\vec{U}, \vec{Z}]$. For example, the log full conditional density of the dependence parameters is (up to an additive constant)
\begin{equation}
    \frac{1}{2}\left[
    |\mat{P}| -
    \vec{u}^\trans \mat{G} \vec{u} - 2 \vec{u}^\trans \mat{Q} (\vec{z} - \mat{X}_{(-z)} \vec{\gamma}) - 
    (\vec{z} - \mat{X}_{(-z)} \vec{\gamma})^\trans \mat{H} (\vec{z} - \mat{X}_{(-z)} \vec{\gamma})
    \right] + \log p(\tau_U, \tau_Z, \phi_U, \phi_Z, \rho),
\end{equation}
where $\mat{P}$, $\mat{G}$, $\mat{H}$, and $\mat{Q}$ are sparse matrices depending on $(\tau_U, \tau_Z, \phi_U, \phi_Z, \rho)$, and $p(\tau_U, \tau_Z, \phi_U, \phi_Z, \rho)$ is the prior for these parameters. In contrast, the REML approach requires inverting $\mat{G}$, $\mat{R}$, $(\mat{H}-\mat{Q}^\trans\mat{G}^{-1}\mat{Q})$, $\mat{M}$, and $\mat{C}^\trans \mat{M} \mat{C}$ in \cref{eq:rl-components} and \cref{eq:res-lik} at each evaluation of the restricted likelihood.

Another advantage of the Bayesian approach is that it is easier to incorporate non-Gaussian exposures by distinguishing between $\vec{Z}$ in \cref{eq:model-joint} and the exposure parameterization in the outcome model.
For example, in our analysis of the food access data in \cref{sec:example}, we replace $Z$ with $\log Z$ in \cref{eq:model-joint} to make the assumption of joint normality more plausible, while using $Z$ in the outcome model to retain the desired interpretation of regression coefficients on the original percentage point scale. Of course, that comes with the caveat that our assumptions are now based on the transformation of the exposure variable.
The exposure model and outcome model may be further decoupled by, e.g., assuming joint normality of $\vec{U}$ and a latent variable in a probit model of a binary exposure.

\subsection{A regularization prior on the precision matrix of \protect\texorpdfstring{$(\vec{U}, \vec{Z})$}{$(U, Z)$}}
\label{sec:regularization}

The estimation of the joint precision matrix $\mat{P}$ of $(\vec{U}, \vec{Z})$ is critical in mitigating bias due to the unobserved spatial confounder $\vec{U}$.
However, the present setting is a ``low-information'' one, as we neither observe $\vec{U}$ directly nor obtain independent replicates.
In such settings, the restricted likelihood may have maxima at the boundary of allowed values.
For example, \cite{chung2013avoiding} noted that it is not unusual in random effects meta-analysis for the REML estimate of the between-study standard deviation to be zero, and suggested regularizing the REML estimate by multiplying the restricted likelihood by a weakly-informative gamma prior for the between-study variance.
Along another thread, \cite{won2013condition} considered estimating the covariance matrix in high-dimensional settings where maximum likelihood estimates of such covariance matrices are often ill-conditioned and cannot be inverted accurately.
They propose a constrained maximum likelihood approach using the constraint $\kappa(\mat{\Sigma}) \leq \kappa_{max}$, where $\kappa(\mat{\Sigma})$ is the condition number (the ratio of the largest to smallest eigenvalue) and $\kappa_{max}$ is pre-specified.
They note that this optimization is equivalent to maximizing the likelihood times an exponential prior on $\kappa(\mat{\Sigma})$ left-truncated at 1.

We have observed that this problem manifests in the fully Bayesian implementation as occasional failure of the MCMC sampler to converge.
For our purposes, we adopt a truncated exponential prior for $\kappa(\mat{P})$ with  rate 1/10 and range $(1, \infty)$ which directly addresses the ill-conditioning problem.
With this specification, the difference in log prior density between $\kappa(\mat{P}) = 1$ and $\kappa(\mat{P}) = 100$ is $9.9$.

\section{Simulation study}
\label{sec:simulation}

\subsection{Linear effect}

We perform simulations to compare the affine estimator to the non-spatial and spatial random effect estimators under several generative models (GMs).
Under all GMs, we consider a single measured covariate, $X$, generated uniformly on $(-1/2, 1/2)$. We assume that the mean of $\vec{Z} | \mat{X}_{(-z)}$ is $\vec{X}$.
Four GMs reflect $(\vec U, \vec Z) | \mat{X}_{(-z)}$ generation according to \cref{eq:model-joint} and \cref{eq:q-spec}, with $\vec{U} | \vec{Z}, \mat{X}_{(-z)}$ and $\vec{Z} | \vec{U}, \mat{X}_{(-z)}$ being one-dimensional CAR models.
The within-variable dependence parameters are denoted by $\phi_U$ and $\phi_Z$, and precision parameters by $\tau_U = \tau_Z = 1$.
The first model we consider is the unconfounded GM (GM 1), where $\vec{U}$ and $\vec{Z}$ are independent of each other ($\rho = 0$), but still spatially structured with $\phi_U = 0.5$ and $\phi_Z = 0.2$.
The unmeasured variable $\vec{U}$ is still predictive of the outcome, hence inducing spatial correlation in the observed outcomes.
For the remaining three CAR models, we specify cross-variable dependence ($\rho = 0.3$), and vary the within-variable dependence parameters $(\phi_U, \phi_Z)$ at $(0.5, 0.2)$ for GM 2, representing a confounder at a larger spatial scale than the exposure, $(0.2, 0.5)$ for GM 3, representing a confounder at a smaller spatial scale than the exposure violating our causal assumptions, and $(0.35, 0.35)$ for GM 4, where confounder and exposure vary at the same spatial scale.
The fifth and sixth GMs represent situations in which the analysis model is mis-specified.
For GM 5, $\vec U$ was generated such that its \textit{marginal} distribution  is a one-dimensional CAR model with $\phi_U = 0.5$ and $\tau_U = 1$, and $\vec{Z} | \vec{U}, \mat{X}_{(-z)} \sim \N[\vec{U} + \vec{X}, \mat{I}]$. Therefore, in this GM, the model \cref{eq:model-joint} is mis-specified in that $\mat{G}$ does not describe the true precision matrix of $\vec{U} | \vec{Z}, \mat{X}_{(-z)}$, and the assumption of constant conditional correlation is violated.
However, the precision matrix of $\vec{Z} | \vec{U}, \mat{X}_{(-z)}$ is still correctly specified, and the cross-Markov property holds.
For GM 6, $\tan(U)$ takes the place of $U$ in \cref{eq:model-joint}, so that the joint normality assumption on $(\vec{U}, \vec{Z}) | \mat{X}_{(-z)}$ is violated.
In all six GMs, the potential outcomes are generated as $Y_i(z) \overset{iid}{\sim} \mathrm{Poisson}\left[\exp\left\{z + X_i + U_i \right\}\right]$.

\begin{table}[!b]
\small
\caption{Simulation results from 500 data sets of size $n = 300$. The -RS suffix indicates estimators with the restriction $\phi_Z \leq \phi_U$. \vspace{8pt}}
\label{tab:sim}
  \centering
  \begin{tabularx}{0.9\textwidth}{rlrd{-2}d{-2}d{-2}d{-2}}
    \toprule
    &
    \multicolumn{1}{l}{Mechanism} & \multicolumn{1}{c}{Estimator}  & \multicolumn{1}{P{5em}}{Bias} & \multicolumn{1}{P{5em}}{Std. Err.} & \multicolumn{1}{P{5em}}{RMSE} & \multicolumn{1}{P{5em}}{95\% CI \newline Coverage} \\
    \midrule
    GM 1 & Unconfounded
 & Non-spatial       & 0.00 & 0.14 & 0.14 & 0.62 \\
 & & Spatial       & 0.01 & 0.12 & 0.12 & 0.93 \\
 & & Spatial-RS    & 0.02 & 0.12 & 0.12 & 0.93 \\
 & & Affine    & 0.03 & 0.27 & 0.27 & 0.98 \\
 & & Affine-RS & 0.03 & 0.29 & 0.29 & 0.96 \\
            \addlinespace
            \addlinespace
    GM 2 & Large-scale 
            & Non-spatial       & 0.37 & 0.15 & 0.40 & 0.02 \\
 & confounder & Spatial       & 0.37 & 0.12 & 0.39 & 0.05 \\
 & & Spatial-RS    & 0.37 & 0.12 & 0.39 & 0.04 \\
 & & Affine    & 0.24 & 0.34 & 0.41 & 0.94 \\
 & & Affine-RS & 0.14 & 0.30 & 0.33 & 0.95 \\
            \addlinespace
            \addlinespace
    GM 3 & Large-scale 
            & Non-spatial       & 0.33 & 0.14 & 0.36 & 0.03 \\
 & exposure & Spatial       & 0.35 & 0.11 & 0.36 & 0.03 \\
 & & Spatial-RS    & 0.34 & 0.10 & 0.35 & 0.04 \\
 & & Affine    & 0.33 & 0.25 & 0.42 & 0.91 \\
 & & Affine-RS & 0.23 & 0.36 & 0.42 & 0.83 \\
            \addlinespace
            \addlinespace
    GM 4 & Same scales
            & Non-spatial       & 0.34 & 0.15 & 0.38 & 0.03 \\
 & & Spatial       & 0.35 & 0.11 & 0.37 & 0.05 \\
 & & Spatial-RS    & 0.35 & 0.11 & 0.37 & 0.06 \\
 & & Affine    & 0.29 & 0.27 & 0.39 & 0.93 \\
 & & Affine-RS & 0.18 & 0.32 & 0.36 & 0.89 \\
            \addlinespace
            \addlinespace
    GM 5 & Non-constant
            & Non-spatial       & 0.35 & 0.11 & 0.37 & 0.00 \\
 & conditional & Spatial       & 0.37 & 0.06 & 0.37 & 0.00 \\
 & correlation & Spatial-RS    & 0.37 & 0.06 & 0.37 & 0.00 \\
 & & Affine    & 0.38 & 0.27 & 0.47 & 0.57 \\
 & & Affine-RS & 0.22 & 0.20 & 0.30 & 0.75 \\
\addlinespace
\addlinespace
GM 6 & Non-normal
& Non-spatial       & 0.24 & 0.08 & 0.26 & 0.06 \\
 & joint & Spatial       & 0.24 & 0.08 & 0.26 & 0.14 \\
 & distribution & Spatial-RS    & 0.24 & 0.08 & 0.26 & 0.15 \\
 & & Affine    & 0.18 & 0.18 & 0.26 & 0.94 \\
 & & Affine-RS & 0.08 & 0.17 & 0.19 & 0.96 \\
    \bottomrule
  \end{tabularx}
\end{table}

Under each GM we generate 500 data sets of size $n = 300$ and fit the non-spatial, spatial, and affine estimators.
When assumptions on the forms of variances are required, we assume CAR structures, and for the affine estimator we assume that $\mat{Q}$ is of the Markov form \cref{eq:q-spec}.
Linear predictor models for the exposure and outcome are correctly specified.
For the spatial and affine estimators, we evaluate variations with and without the restriction that $\phi_Z < \phi_U$ discussed in \cref{subsec:scale_restriction}, with the restricted estimators denoted by (-RS).
% only attempting to adjust for variation in the confounder at a larger spatial scale than that of the treatment.
For regression coefficients we used Gaussian priors with mean 0 and standard deviation 10.
For the Spatial-RS, Affine, and Affine-RS estimators we used the regularization prior discussed in \cref{sec:regularization}.
This prior couples the distributions  of $\vec{U}$ and $\vec{Z}$ which is not usually a feature in spatial analyses, and so is not used for the Non-spatial and Spatial estimators.
Instead, for the Non-spatial and unrestricted Spatial estimators flat priors were used for all variance parameters.
In all cases the precision matrix $\mat{P}$ was restricted to be positive definite.
Due to the computational cost of computing the condition number of the precision matrix $\mat{P}$ when evaluating the prior for variance parameters, we used an approximation to the condition number acquired by the same model and parameter values on a four-location, one-dimensional ring instead of the 300-location line (the first and last locations are also neighbors).
Posterior samples were drawn using 10,000 Gibbs sampler iterations after 1000 burn-in iterations.

Table~\ref{tab:sim} displays the simulation results in terms of bias, standard deviation, and root mean squared error (RMSE) of posterior means across data sets, and empirical coverage of 95\% equal-tail credible intervals.
We first note that there is minimal difference between the Spatial and Spatial-RS estimators, even when the spatial scale assumption is violated in the large-scale exposure scenario.
Additionally, the Spatial and Spatial-RS estimators have similar biases to the Non-spatial estimator.
For the unconfounded GM 1, all estimators are unbiased, all spatial estimators have approximately correct confidence interval coverage, but both affine estimators have much larger standard errors and therefore RMSE.
As expected due to misspecification of the dependence structure, the posterior distributions from the Non-spatial models are too concentrated and therefore the credible intervals are anti-conservative.
In all GMs with confounding (2--6), the Affine-RS estimator mitigates bias relative to the Non-spatial estimator, whereas the Spatial and Spatial-RS estimators do not.
The unrestricted Affine estimator generally mitigates bias to a lesser extent, especially in the large-scale exposure, same-scales, and non-constant conditional correlation scenarios.
When the restricted scale assumption is correct, the Affine-RS estimator has a smaller standard error than the unrestricted Affine estimator.
Additionally, the Affine-RS estimator has a smaller RMSE than all other estimators except in the case of a large-scale exposure where its scale restriction is false.
In the presence of unobserved confounding, both affine estimators have credible interval coverage rates far superior to the other estimators.
Both have approximately nominal coverage rates in the independent, large-scale confounder, and non-normal joint distribution scenarios.
Analogous simulation results for maximum a posteriori (restricted likelihood multiplied by a prior) estimation in the Gaussian outcome are available in Appendix~\ref{app:sim-map}.

Although our model specifies that $\vec{U}$ is independent of $\vec X$, we conducted a smaller simulation for a scenario in which the large-scale unmeasured confounder in GM 2 is correlated with the measured confounder, by specifying that $\E[\vec U | \vec X] = \vec X$.
We simulated 100 data sets and fit the same models (i.e., without a mean model for $\vec{U}$) as in the previous simulations.
The results in estimating $\beta_z$ were similar to those under GM 2 (Appendix~\ref{app:sim-xucorr}) but the estimates of $\beta_x$ were biased upward (not shown). This is expected since the part of $U$ that is correlated with $X$ is captured and adjusted for with the inclusion of $X$ in the outcome model, and the affine estimator targets the component of $U$ that is orthogonal to $X$.

\subsection{Nonlinear effect}

The bias-variance trade-off observed between the spatial and affine estimators in the linear case was also observed for a non-linear effect.
We generated 500 data sets of size 300 where $\vec U, \vec Z$ are generated from \cref{eq:model-joint} with $(\tau_U, \phi_U, \tau_Z, \phi_Z, \rho) = (1, 0.5, 1, 0.2, 0.3)$, and $Y$ is a Poisson variable with log link and linear predictor in the form of the right-hand side of \cref{eq:dgm-nonlinear}, with $(\beta_0, \beta_x) = (0, 0)$ and $f(z) = 2 / (1 + e^{-6z}) - 1$.
Therefore, the true effect curve is an anti-symmetric sigmoid with asymptotes $-1$ and $1$.
We fit the restricted-scale semiparametric spatial and constrained affine estimators using a penalized cubic spline model with a radial basis and used the same priors as in the linear simulations.
Inference was based on 10,000 posterior draws after 5,000 burn-in iterations.

Figure~\ref{fig:sim-nonlinear} displays a graphical summary of the simulation results.
For the most part, both estimators capture the general shape of the mean response curve.
However, the spatial estimator is biased toward more extreme estimates as the exposure deviates from 0, and this bias is mitigated by the constrained affine estimator.
On the other hand, the constrained affine estimator exhibits substantially greater variability, especially for exposure ranges with limited available data (away from an exposure value of 0).
There also appears to be an asymmetry in that when the true log PAERC is negative the bias of the spatial estimator is more pronounced and the constrained affine estimator is more effective at mitigating bias than when the log PAERC is positive.
This asymmetry is likely due to the non-linear relationship between the model's linear predictor and the expected outcome according to the Poisson likelihood.

\begin{figure}[!t]
	\centering
	\includegraphics[scale=0.8]{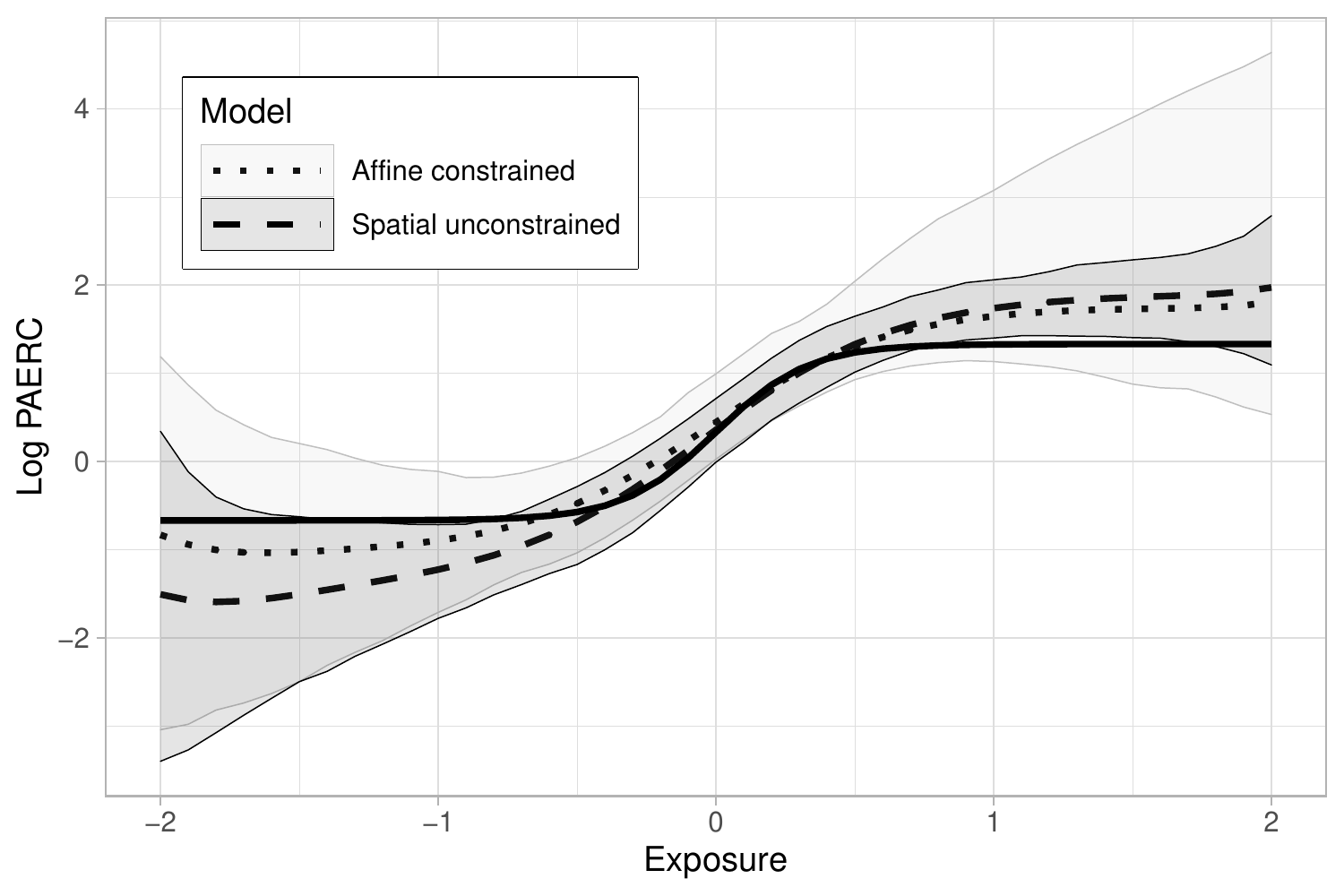}
	\caption{Nonlinear effect simulation results. Mean and pointwise 95\% sampling intervals of posterior mean log population average exposure-response curve from 500 data sets of size $n = 300$. True log PAERC indicated by solid black curve.}
	\label{fig:sim-nonlinear}
\end{figure}

\section{Estimating the county-level effects of poor supermarket availability on CVD mortality}
\label{sec:example}

Here, we use the affine estimator in order to estimate the county-level effect of poor supermarket availability on CVD mortality. We consider the affine estimator with and without the spatial scale restriction, and we also consider the spatial random effect estimator with the spatial scale restriction, an extension  to the results shown in \cref{sec:data}. In the outcome model, we include the exposure on the percentage point scale to aid in interpretation of its coefficient. However, we replace $\vec{Z}$ with $\log \vec{Z}$ in the joint model \cref{eq:model-joint} to better satisfy the condition of joint normality.
For the Spatial-RS model we restrict $\rho$ to be zero, and for the Spatial-RS and Affine-RS models we apply the constraint $\phi_U \geq \phi_Z$. Note that the prior distribution for the Spatial-RS model differs substantially from that of the unconstrained Spatial model.
In all three models we use a similar approximation to the condition number prior of \cref{sec:regularization} that was used in the simulation study: rather than computing the condition number on the full joint precision matrix $\mat{P}$, we use the analog of $\mat{P}$ derived from a $4 \times 4$ regular grid.
Posterior distributions from all models were simulated by retaining 10,000 Gibbs sampler iterations after 1,000 burn-in iterations.
The affine-RS model took approximately 1.5 hours to fit on a laptop for a sample size of $n = \numcounties$.

In \cref{sec:example-assumptions} we examine the assumptions underlying the affine estimator, and in \cref{sec:example-results} we report summaries of the posterior distribution of the causal effect estimates.

\subsection{Examining the plausibility of the assumptions in the context of our study}
\label{sec:example-assumptions}

A number of assumptions, previously presented in Table~\ref{tab:assumptions}, are necessary to identify the causal effect of interest in the presence of unmeasured spatial confounders.

Temporal ordering and SUTVA are standard assumptions in causal inference and are necessary to define our causal effect.
The temporal ordering of the exposure and outcome is immediate satisfied since the exposure data were compiled from 2000 and 2006 data sets, while the mortality outcome data were compiled from 2007 reports.
SUTVA is expected to hold, at least approximately, since we can assume that the county-level effect of poor supermarket access on CVD mortality is due to individual-level causal effects, and the home address county listed on death certificates in 2007 corresponds well to the deceased person's county of residence from 2000 to 2006.

Outcome additivity and the appropriateness of the joint normality and CAR assumptions for $(\vec{U}, \vec{Z})|\mat{X}_{(-z)}$ are modeling assumptions that may be at least partially addressed via standard diagnostics.
Maps of Pearson residuals based on posterior mean parameters and plots of those residuals versus linear predictors indicated no visually apparent residual spatial correlation or non-linearity in either the log exposure or outcome models in the affine-RS approach.
A scatterplot of the joint distribution of the residual log exposure after adjusting for covariates versus mean imputed confounder $\vec{U}$ appeared Gaussian.

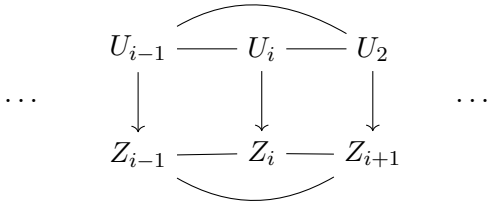
\begin{figure}[!b]
	\centering
	\begin{tikzpicture}
	% nodes %
	\node[text centered] (U1) {$U_{i-1}$};
	\node[right = 0.8 of U1, text centered] (Ui) {$U_i$};
	\node[right = 0.8 of Ui, text centered] (U2) {$U_2$};
	
	\node[left = 0.4 of U1] (useless1) {};
	\node[left = 0.2 of useless1] (useless11) {};
	\node[right = 0.4 of U2] (useless2) {};
	\node[right = 0.2 of useless2] (useless21) {};
	\node[below = 0.4 of useless11, text centered] (dots1) {$\dots$};
	\node[below = 0.4 of useless21, text centered] (dots2) {$\dots$};

	\node[below = 0.8 of U1, text centered] (Z1) {$Z_{i-1}$};
	\node[below = 0.8 of Ui, text centered] (Zi) {$Z_i$};
	\node[below = 0.8 of U2, text centered] (Z2) {$Z_{i + 1}$};
	
	% edges %
	\draw[->] (U1) -- (Z1);
	\draw[->] (U2) -- (Z2);
	\draw[->] (Ui) -- (Zi);
	\draw[-] (U1) -- (Ui);
	\draw[-] (U2) -- (Ui);
	\draw[-] (U1) to [out=30, in=150] (U2);
	\draw[-] (Z1) -- (Zi);
	\draw[-] (Z2) -- (Zi);
	\draw[-] (Z1) to [out=-30, in=-150] (Z2);
	\end{tikzpicture}
\caption{Graph representation of the cross-Markov property $p(Z_i | \vec{Z}_{-i}, \vec{U}) = p(Z_i | \vec{Z}_{-i}, U_i)$.}
\label{fig:cross_markov_graph}
\end{figure}

The cross-Markov and constant conditional correlation assumptions are assumptions about the relationship between unobserved confounders and the exposure, to which standard diagnostics are not applicable. The plausibility of these assumptions depends heavily on the application and hypothesized confounders. For example, consider an unmeasured variable representing cultural preference toward purchasing prepared food from restaurants versus cooking at home. This variable might act as a confounder in our study, since an increase in such a preference could both depress the demand for and availability of supermarkets, and might drive food choices independent of supermarket availability. The cross-Markov property for this variable (an illustration of which is shown in Figure~\ref{fig:cross_markov_graph}) allows for such cultural preferences to have complex dependence structures across locations. However, grocery store accessibility within a county $i$, $Z_i$, is only allowed to depend on such cultural preferences only through its value within the county (conditional on the grocery accessibility in all other locations). This assumption is reasonable for large counties where the food culture in neighboring counties does not directly influence the demand for (and eventual availability of) supermarkets except through its correlation with the food culture within the county itself.
Given the cross-Markov property, the constant conditional correlation assumption implies that the strength of the relationship between this aspect of food culture and supermarket availability is constant (conditional on other observed variables in $\mat{X}_{(-z)}$).
This can be seen by noting that
\begin{equation}
\label{eq:constant-conditional-correlation-regression}
    \E[Z_i | \vec{Z}_{-i}, U_i] = \vec{x}_{(-z), i}^\trans \vec{\gamma} + \frac{\phi_Z}{|\partial_i|} \sum_{j \in \partial_i} (Z_j - \vec{x}_{(-z),j}^\trans \vec{\gamma}) + \rho \sqrt{\frac{\tau_U}{\tau_Z}} U_i.
\end{equation}
Thus if both the cross-Markov and constant conditional correlation hold, $\sqrt{\tau_U / \tau_Z} U_i$ above behaves like an additive predictor of $Z_i$ with regression coefficient $\rho$.
A plot of residuals from the regression implied by \cref{eq:constant-conditional-correlation-regression} did not indicate any departures from linearity.

An assumption on the spatial scales of unobserved confounders is also critical for reliable identification of causal effects.
%Although the simulation study indicates that the affine-RS model can somewhat mitigate bias even when the spatial scale of the counfounder is smaller than that of the exposure, it does not do so as effectively and the mean squared error of the estimates suffers.
We can evaluate this assumption within the model by examining the joint posterior distribution of $(\phi_U, \phi_Z)$.
In our case, the posterior distribution of $\phi_U - \phi_Z$ from the unconstrained Affine model was approximately Gaussian with mean $0.091$, standard deviation $0.024$, and first percentile $0.030$.
The posterior from the Affine-RS model was similar, indicating that the assumption is satisfied within the scope of the model.

\subsection{Estimating the effect of poor supermarket availability on CVD mortality}
\label{sec:example-results}

\begin{figure}[!b]
    \centering
    \includegraphics[scale = 0.8]{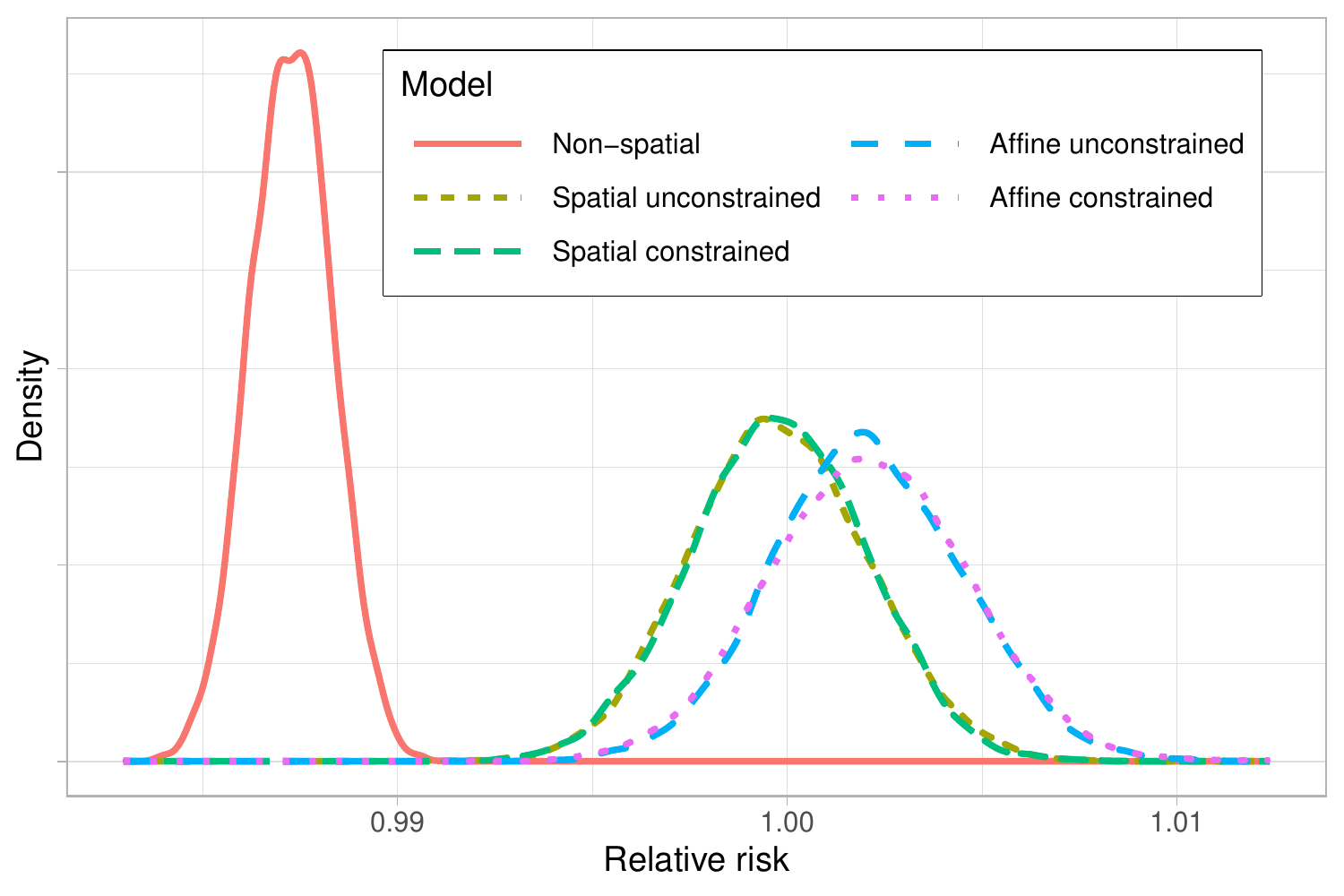}
    \caption{Posterior densities of county-level relative risk of CVD mortality in the 65+ age range due to a 1 percentage point increase in the proportion of households with no vehicle and more than 1 mile from a supermarket or large grocery store.}
    \label{fig:posterior_estimates}
\end{figure}

Figure~\ref{fig:posterior_estimates} displays the posterior distribution of the exponentiated exposure coefficient from the Non-spatial, Spatial, Spatial-RS, Affine, and Affine-RS models. As noted in \cref{subsec:other_outcomes}, this quantity can be interpreted as the relative expected risk of CVD mortality among the population in the 65+ age range due to a one percentage point increase in poor supermarket access in a randomly-chosen county.
%Here, this causal effect is presented in terms of the county-level relative risk of CVD mortality in the 65+ age range due to a 1 percentage point increase in the proportion of households in the county with no vehicle and more than 1 mile from a supermarket or large grocery store.
The agreement of the posterior distribution of $\phi_U-\phi_Z$ in the spatially restricted and unrestricted models implies that the posterior densities from the Spatial and Spatial-RS models closely coincide (both posterior geometric means $0.999$, 95\% CIs $(0.988, 1.011)$), as do those from the affine and affine-RS models (posterior geometric mean $1.005$, 95\% CI $(0.993, 1.018)$ and $1.005$ $(0.992, 1.018)$, respectively).
As we saw in \cref{sec:data}, the non-spatial model reports a definitive, protective effect of poor supermarket availability (posterior mean $0.968$, 95\% CI $(0.962, 0.973)$). In contrast, all four spatial models return a smaller and potentially null effect. The Spatial and Spatial-RS models return approximately null effects with posterior probability of a relative risk smaller than 1 equal to 0.54 and 0.53, respectively. In contrast, the Affine and Affine-RS models estimate that poor supermarket access might have a harmful effect on CVD mortality, with posterior probability of a relative risk greater than 1 equal to $0.8$ and $0.79$, respectively.

The change in the point estimate between the spatial random effect models and the affine models is largely attributable to the posterior distribution of $\rho$. This distribution is skewed slightly left, with a posterior mean of $-0.020$ and 95\% CI $(-0.040, -0.001)$, indicating confounding by the latent variable $\vec{U}$, even though the posteriors of $\exp(\beta_Z)$ from the Spatial(-RS) and Affine(-RS) models overlap. Among posterior draws from the Affine-RS model, the correlation between $\rho$ and $\exp(\beta_Z)$ is $-0.448$, indicating that unobserved spatial confounding is likely to bias the spatial estimator downwards.

\begin{figure}[!t]
    \centering
    \includegraphics[scale = 0.8]{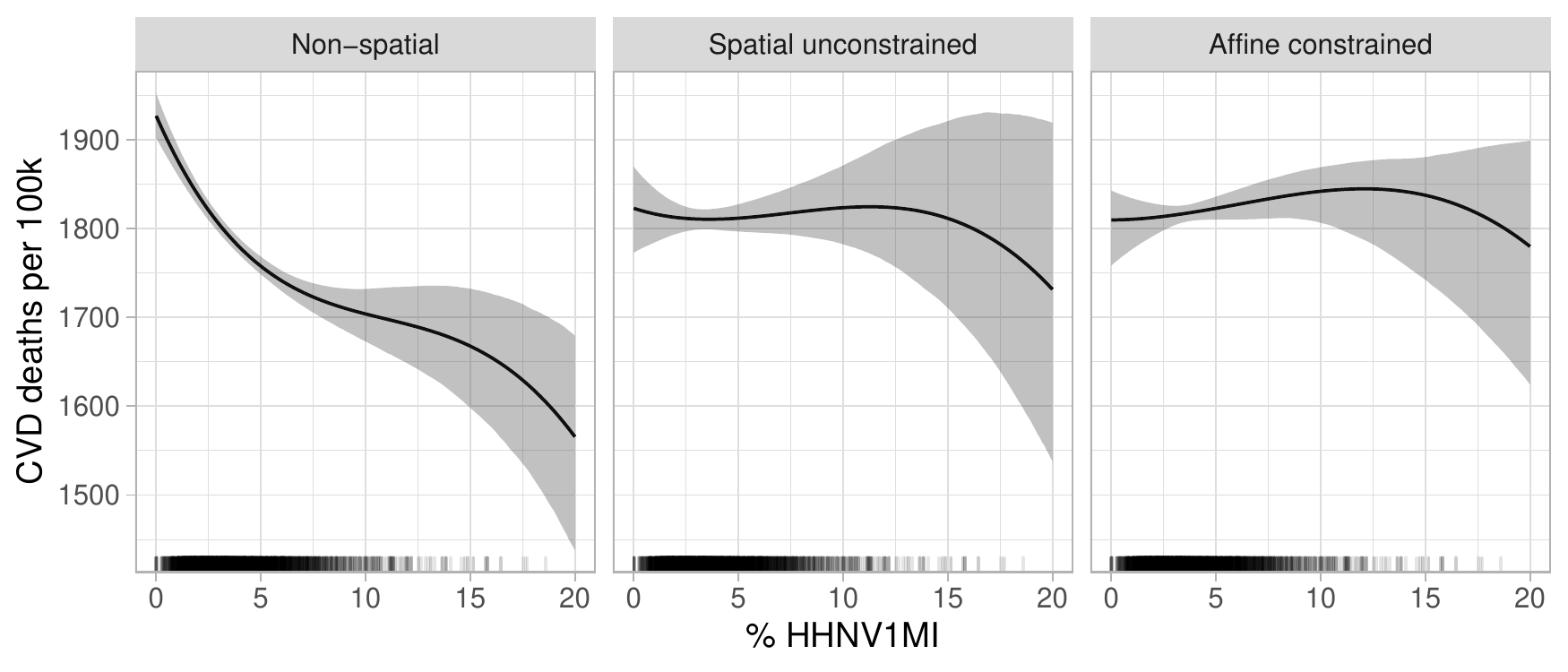}
    \caption{Posterior geometric mean and pointwise 95\% credible bounds for the population average exposure-response curve.
    Rugs indicate observed exposures.
    Results from the constrained spatial and unconstrained affine models are similar to those from the unconstrained spatial and constrained affine models, respectively.}
    \label{fig:posterior_bands}
\end{figure}

We also fit semiparametric versions of each model following the approach in \cref{sec:semiparametric}.
Due to the non-uniform distribution of exposures, we used a truncated cubic basis for the penalized spline.
We retained 10,000 MCMC iterations after 5,000 burn-in iterations.
The affine-RS model took approximately 2 hours to fit.
Results are shown in \cref{fig:posterior_bands} and are in broad agreement with those from the generalized linear models.
The non-spatial model indicates a protective effect of increasing the proportion of households with poor supermarket access on CVD mortality across the observed exposure range. In contrast, the spatial and affine models indicate null and weakly harmful effects, respectively.
In all models, the posterior geometric mean of PAERC indicates a protective effect at extremely high levels of poor supermarket access (above 15\%), though this is likely an artifact of extrapolation of the spline basis since data are sparse in that range, and the credible bands widen dramatically.

\section{Discussion}
\label{sec:discussion}

By positing a joint model for the exposure and unmeasured spatially-correlated confounders, we were able to extend commonly-used spatial data analysis tools to mitigate bias due to such confounders.
In contrast to existing approaches that indicate a potentially protective effect of poor supermarket access on cardiovascular disease deaths among the elderly at the county level,  the proposed approach leads to estimated effects and inferences that are more plausible and in line with subject-matter knowledge, which indicate that poor supermarket access is likely to be harmful on cardiovascular health.
%This estimated harmful effect is consistent with the prevailing conceptual model of food access and cardiovascular disease, and with the limited evidence at the individual-level.
Our study also contributes to the literature on racial and socio-economic disparities.
Much of the recent research on food access and health makes reference to \emph{food deserts}, census tracts with low income and poor food access \citep{usda2012food}.
At the neighborhood level, supermarkets are more prevalent in high-income neighborhoods and in neighborhoods with predominantly white residents compared to those whose residents have lower incomes or are black \citep{morland2002neighborhood}.

At the same time, there exist limitations to our study that extend outside the scope of unobserved confounders.
County-level effects of supermarket access can be extremely relevant for policy making and local planning, but they cannot be directly translated to individual-level effects that may also be of substantial scientific interest.
Additionally, although outcomes were measured in the 65+ age range, the exposure and covariates were generally measured across all age ranges.
%However, the county-level causal effect of the exposure across all age groups on mortality in the 65+ group is well-defined and relevant even though it would be less so at the individual level.
Finally, it is unlikely that effects of poor supermarket access would manifest in differences in mortality in the short-term (e.g., the following year).
Instead, effects are more likely to be cumulative over timespans on the order of many years.
%In designing and interpreting our analysis, we made the assumption that the food access summary the year before the mortality data were collected would be the best available estimate of food access over longer preceding periods.

From a statistical perspective, our approach to mitigating bias from unobserved spatial confounders is rooted in the causal inference framework and exploits spatial statistics tools that can be used to \emph{directly adjust} for structured unmeasured confounding.
We hope that our work contributes to the growing bridge between spatial data analysis and causal inference. The methodology is intended to be amenable to researchers accustomed to the usual spatial statistics literature, but could potentially be useful in situations calling for mixed models more generally, with appropriate modifications.
For that reason, the proposed approach may be widely applicable to scenarios in spatial statistics, time-series analyses, and spatio-temporal settings.

One of the key assumptions in drawing causal conclusions is that of positivity. In the presence of spatial confounders, and for positivity to hold, the spatial scale of the confounder must be larger than that of the treatment. Checking the robustness of estimated effects using the affine estimator with and without the spatial scale restriction can provide intuition on the plausibility of positivity due to unmeasured spatial confounders. This should, however, be employed with care, if spatial mediating variables are expected to exist.

A natural question that arises is whether, and at what occasions, methodology that directly adjusts for unmeasured confounding should be preferred over classical sensitivity analysis. We find that, in settings with structured data, such as spatial and temporal data, unmeasured confounders will also be expected to be structured. In those situations, we find that methodology that directly adjusts for these variables can provide more accurate effect estimates, and strengthen the causal claim of scientific conclusions. An interesting line of future research could extend sensitivity analysis to settings where methods like the one presented here alleviate bias from structured confounders, and sensitivity of results in the presence of unstructured covariates is evaluated.

%The key novelty is the incorporation of dependence between the spatial random effect and the exposure of interest, thus treating the spatial random effect as a confounder rather than an additional variance component.

Based on the structural model \cref{eq:dgm-general}, we discussed a set of assumptions that allowed for identification of the bias correction term using only the observed data, while providing a simple expression for the expected value of the confounder conditional on the exposure: joint normality of the spatial confounder and exposure of interest, the cross-Markov and constant conditional correlation assumptions.
There are several ways in which these assumptions could be relaxed.
First, joint normality may be relaxed by assuming joint normality of an underlying random effect process, with the realizations of both the exposure and covariates arising from other distributions.
For example, a latent probit model could be used to assess the effects of binary exposures.
Furthermore, it may be possible to use more flexible, nonparametric alternatives to the multivariate normal form of the random effect structure, such as spatial Dirichlet processes, both in point-referenced \citep{gelfand2005bayesian} and areal \citep{kottas2008modeling} data.
However, questions of identifiability in less parametric settings will likely be difficult to answer.
The cross-Markov assumption may be relaxed by, for example, treating the joint distribution of the exposure and confounder as a multivariate conditional autoregressive process \citep{gelfand2003proper} and expanding the allowable neighbor relations. The constant conditional correlation assumption may be relaxed by allowing the conditional correlation to vary smoothly in space or based on the number of neighboring locations. In any case, formal treatment of general requirements on the (a) spatial dependence structure (such as the ring graph in \cref{subsec:identifiability_ring}), including (b) the cross-Markov structure specifying the conditional dependence between the exposure and the unmeasured confounder, and (c) the conditional correlation between the unmeasured confounder and exposure, allowing for identification of the causal effect in the presence of unmeasured confounding is an interesting topic of future research.
On a more technical note, priors that place a positive probability on $\tau_U^{-1} = 0$ could be allowed.

Furthermore, even though the structural model in \cref{eq:dgm-general} allows for arbitrary interactions among the exposure and measured covariates (allowing, for example, for treatment effect heterogeneity), it assumes that there are no interactions between measured and unmeasured covariates. The extent to which this assumption can be relaxed is an interesting line of future work, especially in the light of recent results in causal inference for unmeasured confounding \citep{damour2019multiple,ogburn2019comment}. Therefore, an interesting question that arises is: Can we harvest the spatial information of the data to mitigate bias from unmeasured confounders without imposing structural assumptions?

We consider this to be the most pressing topic of future study: what are the general conditions under which causal effects are identifiable in the presence of unmeasured spatial confounding, and to what extent is bias mitigation robust to model misspecification?
Within the context of each study, researchers would need to verify whether the set of reasonable (within their context) assumptions suffices for identification of the bias correction term, while we have illustrated situations in which the causal effect is and is not identifiable.

\section*{Acknowledgements}
We would like to thank Dr. Lucas R.F. Henneman for his contributions to the original version of the manuscript. We would also like to thank Dr. Jim Hodges for his valuable thoughts and input during various stages of this paper. The work presented here was partially supported by grants UL1TR002733 and R01MH118927 from the National Institutes of Health.

\clearpage

\section*{Appendices}
\allowdisplaybreaks

	\appendix
	\renewcommand{\theequation}{\thesection.\arabic{equation}}
	\renewcommand{\thefigure}{\thesection.\arabic{figure}}

	\numberwithin{figure}{section}
	\numberwithin{table}{section}
	\addcontentsline{toc}{section}{Appendices}
	\renewcommand{\thesection}{\Alph{section}}
	\numberwithin{equation}{section}
	\setcounter{figure}{0}    
	\setcounter{table}{0}

	\section{Additional information on the creation of the data set}
	\label{app_sec:data}
	
	We compile a county-level data set including supermarket availability, cardiovascular mortality, demographic and behavioral information for 3,093 out of 3,109 counties or county equivalents in the continental United States.
	
	\paragraph{Supermarket availability}
	Store data were acquired from the United States Department of Agriculture \citep{ver2009access}. For each county (or county equivalent), supermarket availability is defined as the proportion of housing units that are more than 1 mile from the nearest supermarket or large grocery store and do not have a car, obtained from the Food Environment Atlas, June 2012 release \citep{usda2012food}.
	The data are produced from the United States Department of Agriculture and Economic Research Service June 2009 Report to Congress on access to affordable and nutritious food \citep{ver2009access}.
	The data set was compiled from a 2006 directory of supermarkets and large grocery stores within  the continental United States and 2000 Census data.
	A supermarket or large grocery store was defined as stores that had at least 2 million USD in annual sales and contained all the major food departments found in a traditional  supermarket.

	\paragraph{Cardiovascular mortality data}
	County-level data on cardiovascular mortality data were obtained from the United States Centers for Disease Control and Prevention (CDC) WONDER query system \citep{cdc2018underlying}.
	We obtained county-level population and CVD death counts, both in the 65+ age range, during 2007.
	A CVD death was defined as a death registered with ICD-10 codes I00--I99 (diseases of the circulatory system) as the underlying cause of death.
	Due to privacy constraints, county death counts below 10 are censored by CDC WONDER.
	Death counts were internally standardized to the total CVD death rate in the 65+ population among all United States counties (including those outside of the continental US), which was not influenced by censoring.

	\paragraph{Covariate information}
	Zip code level demographic information was acquired from the 2000 Census. The median household value was set to N/A for 448 (out of 40,646) zip codes for which the Census measurement was equal to zero, but the zip codes population size was positive. We identified the county (or county equivalent) to which a zip code belongs, and set county-level demographic measurements equal to the mean of the zip code-level measurements within the county, excluding missing values.
	Except for the Census information, we also acquired smoking rates from small-area estimates acquired using the CDC Behavioral Risk Factor Surveillance System data \citep{dwyer2014cigarette}.
	All covariates along with short description and descriptive statistics are shown in \cref{app_tab:table1}.

	\paragraph{}
	Sixteen out of 3,109 counties are excluded from the analysis, leaving 3,093.
	Broomfield County, Colorado, having been created in 2001, is missing from the supermarket availability data.
	Loving County, Texas (2000 population 67) does not have a median home value recorded in the census data.
	Covington city, Emporia city, and Lexington city, Virginia are independent cities (not part of the surrounding counties) with populations on the order of 5,000 and do not appear in the census data set.
	Eleven contiguous counties in southwestern Georgia are missing demographic information in the census data set for an unknown reason.
	
	\paragraph{Links to data sources and sets}
	\begin{enumerate}
		\item Mortality data can be acquired from the CDC WONDER system at \url{https://wonder.cdc.gov} by specifying the following query:
		
		\begin{table}[H]
			\centering
			\small
			\begin{tabular}{ll}
				\toprule
				Parameter & Value \\
				\midrule
				Dataset & Underlying Cause of Death, 1999--20017 \\
				ICD-10 Codes & I00--I99 (Diseases of the circulatory system) \\
				Ten-Year Age Groups & 65--74 years; 75--84 years; 85+ years \\
				Year/Month & 2007 \\
				Group By & County \\
				Show Totals & True \\
				Show Zero Values & True \\
				Show Suppressed & True \\
				Calculate Rates Per & $100,000$ \\
				Rate Options & Default intercensal populations for years 2001--2009 (except Infant Age Groups) \\
				\bottomrule
			\end{tabular}
			\caption{CDC WONDER Online Database query, January 9, 2020.}
			\label{tab:my_label}
		\end{table}
		
		\item Food access data can be aquired from USDA Food Environment Atlas at \url{https://www.ers.usda.gov/data-products/food-environment-atlas/}. We used the June 2012 version.

		\item Complete datasets, including remaining county information are available at in the online data supplement.
	\end{enumerate}

	{ \small
		\renewcommand*{\arraystretch}{1.4}
		\begin{longtable}{llccc}
			\caption{Available covariate information. \\  $^\dagger$Log-transformed in analysis} \\
			\label{app_tab:table1}
			Name & Description & Q1 & Median & Q3 \\
			\hline
			PctUrban & Percentage of population in urban areas & $4.2$ & $16$ & $37$\\
			PctWhite & Percentage of white population & $79$ & $93$ & $97$ \\
			PctBlack & Percentage of black population & $0.21$ & $1.1$ & $8.4$ \\
			PctHisp & Percentage of hispanic population & $0.73$ & $1.5$ & $4.2$ \\
			PctHighSchool & Percentage of population that attended high school & $31$ & $36$ & $41$ \\
			MedianHHInc$^\dagger$ & Median household income ($\times 1000$ USD) & $30$ & $34$ & $39$ \\
			PctPoor  &  Percentage of impoverished population & $9.5$ & $13$ & $18$ \\
			PctFemale & Percentage of female population & $49$ & $50$ & $51$ \\
			PctMovedIn5 &  Percentage of population having lived in area for less than 5 years & $36$ & $40$ & $46$ \\
			PctOccupied & Percentage of housing units that are occupied & $80$ & $87$ & $91$ \\
			MedianHValue$^\dagger$  & Median value of owner occupied housing ($\times 1000$ USD) & $54$ & $72$ & $95$ \\
			PopPerSQM$^\dagger$ & Population per square mile & $20$ & $58$ & $240$ \\
			TotPop$^\dagger$ & Total county population ($\times 1000$) & $12$ & $31$ & $94$ \\
			smokerate & Percentage of population that smokes & $23$ & $26$ & $29$ \\
			\hline
		\end{longtable}
	}

	\section{Identifiability of causal estimands}
	\label{app_sec:causal}
	
	Here, we review causal identifiability of $\mu(z) = E[Y(z)]$, the expected value of the potential outcome for a fixed treatment $z$ over some population, for a binary treatment $z \in \mathcal{Z} = \{0, 1\}$. Note that we do not observe $Y(z)$ for everyone and $\mu(z)$ is an expectation including many unobserved quantities.
	
	A causal estimand is referred to as identifiable under a set of assumptions if it can be written as a function of observables. For $\mu(z)$, on set of assumptions is (1) consistency of potential outcomes, (2) positivity and (3) no unmeasured confounding, since
	\begin{align*}
		\mu(z) = E[Y(z)] = E\{ E[Y(z) | W] \} = E\{ E[Y(z) | Z = z, W]  \} = E\{ E[Y | Z = z, W]\},
	\end{align*}
	where the third equation holds because of the no unmeasured confounding assumption, and the fourth equation holds because of the causal consistency assumption. So $\mu(z)$ is written as a function of the \textit{observed} outcomes among those with $Z = z$, and for that reason it is identifiable.
	
	%If $Z$ is continuous, the set of units with observed exposure values $Z = z$ is likely to be empty. In that case, estimation of $\mu(z)$ is often based on some form of outcome model extrapolation \citep{}.
	
	\section{Implications of dependence assumptions on the joint precision matrix}
	\label{app:partial-corr}
	
	Here we show how the cross-Markov and constant conditional correlation assumptions determine the matrix $\mat{Q}$.
	We suppress dependence on $\mat{X}_{(-z)}$ for simplicity.
	Let
	\begin{equation}
		\mat{P}  = \begin{pmatrix}
			\mat{G} & \mat{Q} \\
			\mat{Q}^\trans & \mat{H}
		\end{pmatrix}
	\end{equation}
	be the joint precision matrix of $(\vec{U}, \vec{Z})$.
	Then, since the precision matrix of coordinates of a Gaussian variable conditional on other coordinates is obtained by dropping the rows and columns corresponding to those other coordinates,
	\begin{equation}
		\Var\left[\begin{pmatrix}U_i \\ Z_j \end{pmatrix} \middle| \vec{U}_{-i}, \vec{Z}_{-j}\right]
		= \begin{pmatrix}
			p_{u_i u_i} & p_{u_i z_j} \\
			p_{z_j u_i} & p_{z_j z_j}
		\end{pmatrix}^{-1}
		= \frac{1}{p_{u_i u_i}p_{z_j z_j} - p_{u_i z_j}p_{z_j u_i}}
		\begin{pmatrix}
			p_{z_j z_j} & -p_{z_j u_i} \\
			-p_{u_i z_j} & p_{u_i u_i}
		\end{pmatrix},
	\end{equation}
	Where $i$ is may or may not be equal to $j$, $p_{u_i z_j}$ is the element of $\mat{P}$ in the row corresponding to $u_i$ and column corresponding to $z_j$, and similar for other coordinates.
	In particular,
	\begin{equation}
		\begin{aligned}
			\Var[U_i | \vec{U}_{-i}, \vec{Z}_{-j}]
			&= \frac{p_{z_j z_j}}{p_{u_i u_i}p_{z_j z_j} - p_{u_i z_j}p_{z_j u_i}}, \\
			\Var[Z_j | \vec{U}_{-i}, \vec{Z}_{(-j)}]
			&= \frac{p_{u_i u_i}}{p_{u_i u_i}p_{z_j z_j} - p_{u_i z_j}p_{z_j u_i}}, \\
			\Cov[U_i, Z_j | \vec{U}_{-i}, \vec{Z}_{-j}]
			&= \frac{-p_{u_i z_j}}{p_{u_i u_i}p_{z_j z_j} - p_{u_i z_j}p_{z_j u_i}}, \\
			\Cor[U_i, Z_j | \vec{U}_{-i}, \vec{Z}_{-j}]
			&= \frac{-p_{u_i z_j}}{\sqrt{p_{u_i u_i} p_{z_j z_j}}}.
		\end{aligned}
	\end{equation}
	If $i \neq j$, the cross-Markov assumption $p(Z_i | \vec{Z}_{-i}, \vec{U}) = p(Z_i | \vec{Z}_{-i}, U_i)$ implies that $\Cor[U_i, Z_j | \vec{U}_{-i}, \vec{Z}_{-j}] = 0$, which in turn implies that $p_{u_i, z_j}$, an arbitrary off-diagonal element of $\mat{Q}$, is zero.
	Additionally, the constant conditional correlation assumption $\Cor(U_i, Z_i | \vec{U}_{-i}, \vec{Z}_{-i}, \mat{X}_{(-z)}) = \rho$ implies that $p_{u_i z_i} = -\rho \sqrt{p_{u_i u_i} p_{z_i z_i}}$.

	\section{Restricted likelihood}
	\label{app_sec:restricted_likelihood}

	\subsection{Formalization}
	\label{app_sec:marginal_variances}
	
	For simplicity, we assume no covariates in the exposure model, but covariates may be included in the outcome model by inclusion in $\mat{X}$.
	Based on \cref{eq:model-joint} the covariance matrix is
	\[ \begin{pmatrix} \G & \Q \\ \Q^T & \h\end{pmatrix}^{-1} = \begin{pmatrix}
	\G^{-1} + \G^{-1}\Q(\h - \Q^T\G^{-1}\Q)^{-1}\Q^T\G^{-1} & - \G^{-1}\Q(\h - \Q^T \G^{-1} \Q)^{-1} \\
	- (\h - \Q^{T}\G^{-1}\Q)^{-1}\Q^T\G^{-1} & (\h - \Q^T\G^{-1}\Q)^{-1}
	\end{pmatrix}. \]
	Based on the properties of the multivariate normal distribution,
	$$
	\bm U | \bm Z \sim N(\bm \mu_{\bm U | \bm Z}, \Sigma_{\bm U | \bm Z})
	$$
	where
	\begin{align*}
		\bm \mu_{\bm U | \bm Z} & = - \G^{-1}\Q(\h - \Q^T \G^{-1} \Q)^{-1}(\h - \Q^T\G^{-1}\Q)^{-1} \bm Z  = - \G^{-1}\Q \bm Z \\
		\Sigma_{\bm U | \bm Z} &=
		\G^{-1} + \G^{-1}\Q(\h - \Q^T\G^{-1}\Q)^{-1}\Q^T\G^{-1} - \\
		&\hspace{20pt} \G^{-1}\Q(\h - \Q^T \G^{-1} \Q)^{-1} (\h - \Q^T\G^{-1}\Q) (\h - \Q^{T}\G^{-1}\Q)^{-1}\Q^T\G^{-1} \\
		& = \G^{-1}.
	\end{align*}
	
	\noindent Based on the above, the marginal variance of $\vec{Z}$ is $(\mat{H} - \mat{Q}^\trans \mat{G}^{-1} \mat{Q})^{-1}$, and the $\Var[\vec{Y} | \vec{Z}] = \Var[\vec{U} | \vec{Z}] + \Var[\vec{\epsilon} | \vec{Z}] = \mat{G}^{-1} + \mat{R}^{-1}$. Further, $E[\vec U | \vec Z] = - \mat G^{-1}\mat Q \vec Z$ leading to the following outcome model integrating $\vec{U} | \vec Z$ out:
	\begin{equation}
		\label{app_eq:outcome_model}
		\vec{Y} | \vec{Z}
		\sim \N\big( \mat{X} \vec{\beta} + \E[\vec{U} | \vec{Z}], \Var[\vec{U} | \vec{Z}] + \Var[\vec{\epsilon}]\big)
		= \N \big(\mat{X} \vec{\beta} - \mat{G}^{-1} \mat{Q} \vec{Z}, \mat{G}^{-1} + \mat{R}^{-1}\big).
	\end{equation}
	
	The full data likelihood can be factored as $f(\vec{y}, \vec{u}, \vec{z} | \vec{\beta}) = f(\vec{y} | \vec{u}, \vec{z}; \vec{\beta}) f(\vec{u} | \vec{z}) f(\vec{z})$.
	Using the outcome model in \eqref{app_eq:outcome_model} and defining $\mat{B} = -\mat{G}^{-1} \mat{Q}$, $\mat{A} = \Var[\vec{Z}] = (\mat{H} - \mat{Q}^\trans \mat{G}^{-1} \mat{Q})^{-1}$, and $\mat{V} = \mat{G}^{-1} + \mat{R}^{-1}$, we have
	\begin{equation}
		f(\vec{Y} | \vec{Z}; \vec{\beta})
		\propto |\mat{V}|^{-1/2} \exp\left[ -\frac{1}{2} \left\{ (\vec{Y} - \mat{B} \vec{Z}) - \mat{X}\vec{\beta} \right\}^\trans \mat{V}^{-1} \left\{ (\vec{Y} - \mat{B} \vec{Z}) - \mat{X}\vec{\beta} \right\} \right],
	\end{equation}
	leading to the following restricted likelihood conditional on $\vec{Z}$,
	\begin{equation}
		r(\vec{Y} | \vec{Z})
		\propto \left( |\mat{V}| \cdot |\mat{X}^\trans \mat{V}^{-1} \mat{X}| \right)^{-\frac{1}{2}}
		\exp\left[ -\frac{1}{2}
		(\vec{Y} - \mat{B} \vec{Z})^\trans
		\left\{ \mat{V}^{-1} - \mat{V}^{-1} \mat{X} (\mat{X}^\trans \mat{V}^{-1} \mat{X}) \mat{X}^\trans \mat{V}^{-1}\right\}
		(\vec{Y} - \mat{B} \vec{Z})\right].
	\end{equation}
	Since $f(\vec{Z})$ does not depend on $\vec{\beta}$, we can write the full restricted likelihood as
	\begin{equation}
		\begin{aligned}
			RL
			&= r(\vec{Y} | \vec{Z}) f(\vec{Z}), \\
			&\propto \left[|\mat{V}| \cdot |\mat{A}| \cdot |\mat{X}^\trans \mat{V}^{-1} \mat{X}| \right]^{-1/2} \\
			&\phantom{==}\times \exp\left[-\frac{1}{2} \left\{
			\begin{array}{r}
				(\vec{Y} - \mat{B} \vec{Z})^\trans
				\left( \mat{V}^{-1} - \mat{V}^{-1} \mat{X} (\mat{X}^\trans \mat{V}^{-1} \mat{X})^{-1} \mat{X}^\trans \mat{V}^{-1} \right)
				(\vec{Y} - \mat{B} \vec{Z}) \\
				+ \vec{Z}^\trans \mat{A}^{-1} \vec{Z}
			\end{array}
			\right\}\right].
		\end{aligned}
	\end{equation}
	
	\subsection{Conservative bounds on the conditional correlation}
	\label{app_sec:rho_bound}
	
	In order to ensure positive definiteness of the precision matrix \(\displaystyle 
	\mat{P}  = \begin{pmatrix}
	\mat{G} & \mat{Q} \\
	\mat{Q}^\trans & \mat{H}
	\end{pmatrix} \), $\rho$ has to be constrained. Even though no convenient form of such constraint is available, a conservative one is given by
	\begin{equation}
		\label{eq:rho_constraint}
		|\rho| < \frac{\min[ \min_i \{ \lambda_{\mat{G},i} \}, \min_i \{ \lambda_{\mat{H},i} \} ]}{\sqrt{\max_i \{ g_{ii} h_{ii} \}}},
	\end{equation}
	where $\lambda_{\mat{G},i}$ and $\lambda_{\mat{H},i}$ are the $i$th eigenvalues of $\mat{G}$ and $\mat{H}$, respectively.
	To establish that, let
	\begin{align*}
		\mat{S} = \begin{pmatrix}
			\mat{G} & \mat{0} \\
			\mat{0} & \mat{H}
		\end{pmatrix}
		\quad \textrm{ and } \quad
		\mat{T}
		= \begin{pmatrix}
			\mat{0} & \mat{Q} \\
			\mat{Q} & \mat{0}
		\end{pmatrix},
	\end{align*}
	we constraint
	$\mat S, \mat T$ such that for any vector $\vec{v} \neq \vec{0}$ of length $2n$, $\vec{v}^\trans \mat{P} \vec{v} = \vec{v}^\trans \mat{S} \vec{v} + \vec{v}^\trans \mat{T} \vec{v} > 0$.
	Let $\set{C} = \{\vec{v} : |\vec{v}| = 1\}$.
	It suffices to show that $\min_{\vec{v} \in \set{C}} \vec{v}^\trans \mat{S} \vec{v} > -\min_{\vec{v} \in \set{C}} \vec{v}^\trans \mat{T} \vec{v}$.
	Note that $\min_{\vec{v} \in \set{C}} \vec{v}^\trans \mat{S} \vec{v}$ is the minimum eigenvalue of $\mat{S}$, and similarly for $\vec{T}$. Also, since $\mat S$ is block diagonal, $\min_i \{\lambda_{\set{S}_i}\} = \min[ \min_i \{\lambda_{\set{G}_i}\}, \min_i \{\lambda_{\set{H}_i}\}]$.
	The eigenvalues of $\mat{T}$ are the roots of $|\lambda \mat{I}_{2n} - \mat{T}| = |\lambda \mat{I}_n| \cdot |\lambda \mat{I}_n - \lambda^{-1} \mat{Q} \mat{Q}| = \prod_{i=1}^{n} (\lambda^2 - \rho^2 g_{ii} h_{ii})$, i.e., $\lambda = \pm \rho \sqrt{g_{ii} h_{ii}}$.
	Thus $-\min_{\vec{v} \in \set{C}} \vec{v}^\trans \mat{T} \vec{v} = |\rho| \sqrt{\max_i \{g_{ii} h_{ii}\}}$, and
	\begin{align*}
		|\rho| < \frac{\min[ \min_i \{ \lambda_{\mat{G},i} \}, \min_i \{ \lambda_{\mat{H},i} \} ]}{\sqrt{\max_i \{ g_{ii} h_{ii} \}}}
	\end{align*}
	guarantees positive definite $\mat{P}$.
	
	\subsection{Approximate standard errors accounting for correlation between parameter estimates}
	\label{sec:se}
	
	For the spatial estimator $\widetilde{\vec\beta}$, approximate standard errors are often constructed assuming known variance parameters: \( \displaystyle \widehat{\Var}\big(\widetilde{\vec{\beta}}\big) \approx \big(\mat{X}^\trans \widehat{\mat{V}}^{-1} \mat{X}\big)^{-1}\).
	We do not recommend applying this idea directly to $\widebar{\vec{\beta}}$ due to the fact that the estimates of $\rho$ and $\beta_z$ are strongly correlated. We account for such correlation with a small modification.
	If all variance parameters except $\rho$ are treated as known, then
	\begin{equation}
		\vec{Y} | \mat{X} \sim \N[\mat{X} \vec{\beta} - \rho \mat{G}^{-1} \mat{Q}^* \vec{Z}, \mat{G}^{-1} + \mat{R}^{-1}],
	\end{equation}
	where $\mat{Q}^*$ is diagonal with elements $q_{ii}^* = -\sqrt{g_{ii}h_{ii}}$ known and independent of $\rho$.
	Then, treating $\rho$ exclusively as a coefficient, we can write $\mat{D} = [\mat{X} | -\mat{G}^{-1} \mat{Q}^* \vec{Z}]$ via concatenation, and obtain an estimated variance $\widehat{\Var}\Big[\Big( \widebar{\vec{\beta}}, \widebar{\rho}\Big)\Big] = \Big(\widehat{\mat{D}}^\trans \widehat{\mat{V}}^{-1} \widehat{\mat{D}}\Big)^{-1}$, from which an estimate of the variance of $\widebar{\vec \beta}$ can be acquired. Based on the estimated variance of $\affine$, Wald-type confidence intervals can be obtained.
	
	Approximate standard errors for the semi-parametric estimator may be obtained similarly by augmenting $\mat{M}$ with $-\mat{G}^{-1} \mat{Q}^* \vec{Z}$ and $\vec{\theta}$ with $\rho$.

	\section{Identifiability results}
	
	\subsection{Matrix results supporting identifiability on the ring graph}
	\label{app:identifiability}
	
	We call a matrix \emph{STDC} if it is symmetric, tridiagonal, and the diagonal, subdiagonal, and superdiagonal are all constant vectors, i.e.,
	\begin{equation}
		\mat{S}_n(a, b) =
		\begin{pmatrix}
			a & b & & & \\
			b & a  & b & & \\
			& \ddots & \ddots & \ddots & \\
			& & b & a & b & \\
			& & & b & a
		\end{pmatrix}.
	\end{equation}
	Let $\set{R}_n$ be the ring graph of order $n$, and $\mat{A}_n(\phi)$ be the unscaled CAR precision matrix
	\begin{equation}
		\mat{A}_n(\phi)
		=
		\begin{pmatrix}
			2 & -\phi & & & -\phi \\
			-\phi & 2  & -\phi & & \\
			& \ddots & \ddots & \ddots & \\
			& & -\phi & 2 & -\phi & \\
			-\phi & & & -\phi & 2
		\end{pmatrix}.
	\end{equation}
	That is, $\mat{A}_n(\phi)$ is the matrix $\mat{S}_n(2, -\phi)$ with the upper-right and lower-left entries modified to be $-\phi$.
	
	\begin{lemma}
		\label{lem:car-det-recurrence}
		For an unscaled CAR precision matrix $\mat{A}_n(\phi)$ of a ring graph $\set{R}_n$,
		\begin{equation}
			\label{eq:car-det-recurrence}
			\begin{aligned}
				\det[\mat{A}_n(\phi)]
				&= 2 \det[\mat{S}_{n-1}(2, \phi)] -
				2\phi^2 \left\{
				\det[\mat{S}_{n-2}(2, \phi)]
				+\phi ^ {n - 2}
				\right\}.
			\end{aligned}
		\end{equation}
	\end{lemma}
	
	\begin{proof}
		We use the Laplace expansion for computing the determinant.
		The first term of \eqref{eq:car-det-recurrence} comes directly from the first term of the expansion along the first row.
		The second and final terms (the only other non-zero terms) are equal because the minor matrices are transposes of each other, and their determinants may be computed via Laplace expansion along the first column of whichever matrix has a non-zero entry in the lower-left corner.
		In this latter expansion, one matrix is STDC and the other is triangular with a constant diagonal.
	\end{proof}
	
	\begin{lemma}
		\label{lem:stdc-det}
		For an STDC matrix $\mat{S}_n(2, -\phi)$ with $|\phi| < 1$,
		\begin{equation}
			\begin{aligned}
				\det[\mat{S}_n(2, -\phi)]
				&= 2 \det[\mat{S}_{n-1}(2, -\phi)] - \phi ^ 2 \det[\mat{S}_{n-2}(2, -\phi)], \\
				&= \frac{1}{2 \sqrt{1 - \phi^2}} \left[\left(1 + \sqrt{1-\phi^2}\right) ^ {n + 1} - \left(1 - \sqrt{1-\phi^2}\right)^{n+1}\right].
			\end{aligned}
		\end{equation}
	\end{lemma}
	
	\begin{proof}
		The recurrence relation can be obtained by computing the Laplace expansion along the first row, yielding two non-zero terms.
		The minor matrix in one non-zero term is $\mat{S}_{n-1}(2, -\phi)$, and in the other term the Laplace expansion along the first column has one non-zero term, whose minor matrix is $\mat{S}_{n-2}(2, -\phi)$.
		Initial conditions for the recurrence relation $r(n) = 2 r(n - 1) - \phi^2 r(n - 2)$ can be set to $r(0) = 1$ and $r(1) = 2$ by letting $\mat{S}_{0}(2, -\phi)$ be empty and $\mat{S}_{1}(2, -\phi) = 2$, and the characteristic roots technique yields the solution $r(n) = \frac{1 + \sqrt{1-\phi^2}}{2\sqrt{1-\phi^2}} \left(1 + \sqrt{1 - \phi^2}\right) ^ n - \frac{1-\sqrt{1-\phi^2}}{2\sqrt{1-\phi^2}} \left(1 - \sqrt{1 - \phi^2}\right) ^ n$.
	\end{proof}
	
	\begin{theorem}
		\label{thm:a-det}
		For an unscaled CAR precision matrix $\mat{A}_n(\phi)$ of a ring graph $\set{R}_n$,
		\begin{equation}
			\begin{aligned}
				\det[\mat{A}_n(\phi)]
				&= 
				\frac{1}{\sqrt{1-\phi^2}}\left[
				\left(1 + \sqrt{1-\phi^2}\right)^{n}
				-\left(1 - \sqrt{1-\phi^2}\right)^{n}
				\right] \\
				&\phantom{==}- 2 \left[ \frac{\phi ^ 2}{2\sqrt{1-\phi^2}}\left\{
				\left(1 + \sqrt{1-\phi^2}\right)^{n-1} -
				\left(1 - \sqrt{1-\phi^2}\right)^{n-1} \right\} +
				\phi ^ {n}
				\right].
			\end{aligned}
		\end{equation}
	\end{theorem}
	
	\begin{proof}
		The result follows immediately from Lemma~\ref{lem:car-det-recurrence} and Lemma~\ref{lem:stdc-det}.
	\end{proof}
	
	\begin{lemma}
		\label{lem:cofactors}
		Let $\{C_{ij}^{(n)}\}$ be the cofactors of $\mat{A}_n(\phi)$.
		Then, for $n > 3$,
		\begin{equation}
			\label{eq:cofactors}
			\begin{aligned}
				C_{1, 1}^{(n)} &= \det[\mat{S}_{n-1}(2, -\phi)], & &  \\
				C_{1, j}^{(n)} &= \phi^{j-1} \det[\mat{S}_{n-j}(2, -\phi)] + \phi^{n-j+1} \det[\mat{S}_{j-2}(2, -\phi)], & & j > 1 \\
			\end{aligned}
		\end{equation}
		Note that the second term of $C_{1, j}^{(n)}$ converges linearly with rate $\phi$ to 0 as $n$ increases.
	\end{lemma}
	
	\begin{proof}
		The first cofactor $C_{1,1}^{(n)}$ can be verified by inspection.
		In computing the other cofactors along the first row, for each corresponding minor matrix, the first column has $-\phi$ in the first and final positions, and $0$ elsewhere.
		We examine the minors corresponding to those non-zero positions.
		
		In the first position, each minor matrix is an upper-triangular block matrix: the upper-left $(j-1)\times(j-1)$ block is upper-triangular with $-\phi$ along the diagonal, the lower-left block is $\mat{S}_{n-j}(2, -\phi)$, and the lower-left block is $\mat{0}$.
		Thus the first term in the second line of \eqref{eq:cofactors}.
		
		In the second position, each minor matrix is a lower-triangular block matrix: the upper-left block is $\mat{S}_{j-2}(2, -\phi)$, the lower-right $(n-j+1)\times(n-j+1)$ block is lower-triangular with $-\phi$ along the diagonal, and the upper-right block is $\mat{0}$.
		Thus the second term in the second line of \eqref{eq:cofactors}.
	\end{proof}
	
	When $\mat{A}_n(\phi)$ is a CAR precision matrix, $C_{1,j}^{(n)}$ is proportional to the marginal correlation between locations $1$ and $j$ (due to Cramer's rule).
	The first term in the second line of \eqref{eq:cofactors} can be thought of as representing the correlation due to the ``leftward'' path from location $j$ to location $1$, and the second term as that due to the ``rightward'' path.
	For a fixed $j$, the ``rightward'' path becomes long as $n$ increases, thus the fact that the second term's limit is zero corresponds to the marginal correlation due to the long ``rightward'' path decreasing to 0.

	\begin{theorem}
		For $|\phi| < 1$,
		\begin{equation}
			\lim_{n\to\infty} \left[\mat{A}_{n}(\phi)^{-1}\right]_{i, j}
			= \frac{1}{2\sqrt{1-\phi^2}}\left(\frac{\phi}{1+\sqrt{1-\phi^2}}\right)^{|i-j|}.
		\end{equation}
	\end{theorem}
	
	\begin{proof}
		The case $i = 1$ follows from Theorem~\ref{thm:a-det} and Lemma~\ref{lem:cofactors} via Cramer's rule.
		Other cases follow by noting that $\left[\mat{A}_n(\phi)^{-1}\right]_{ij} = \Cov[Z_i, Z_j]$ and thus depends only on $|i-j|$.
	\end{proof}
	
	\subsection{Lack of identifiability in the absence of spatial structure}
	\label{app:non-spatial}
	
	We illustrate the way in which the proposed method fails to adjust for confounding in the absence of spatial structure. In the interest of simplicity, we do not include measured covariates and $\vec X = (\vec 1 \ | \ \vec Z)$. 
	Suppose that $\mat{G} = \tau_U \mat{I}$, $\mat{H} = \tau_Z \mat{I}$, and $\mat{R} = \tau_\epsilon \mat{I}$ are all scalar matrices.
	We then have $\mat{Q} = -\rho \sqrt{\tau_U \tau_Z}$, and in the notation of \eqref{eq:res-lik}, $\mat{V} = (\tau_U^{-1} + \tau_\epsilon^{-1}) \mat{I}$, $\mat{A} = \tau_Z^{-1} (1 - \rho^2)^{-1} \mat{I}$, and $\mat{B} = \rho \sqrt{\frac{\tau_Z}{\tau_U}} \mat{I}$.
	Note that $\mat{A}$ does not depend on $\tau_U$.
	Various simplifications to \cref{eq:res-lik} are then available:
	\begin{align*}
		|\mat{V}| \cdot |\mat{A}| \cdot |\mat{X}^\trans \mat{V}^{-1} \mat{X}|
		&= (\tau_U^{-1} + \tau_\epsilon^{-1})^{n-p} \left[ \tau_Z^{-1} (1 - \rho^2)^{-1} \right]^{n} \left[ |\mat{X}^\trans \mat{X}| \right], \\
		\mat{V}^{-1} - \mat{V}^{-1} \mat{X} (\mat{X}^\trans \mat{V}^{-1} \mat{X})^{-1} \mat{X}^\trans \mat{V}^{-1}
		&= (\tau_U^{-1} + \tau_\epsilon^{-1})^{-1} [\mat{I} - \mat{X}(\mat{X}^\trans \mat{X})^{-1} \mat{X}^\trans ], \\
		\vec{Z}^\trans \mat{A}^{-1} \vec{Z}
		&= \tau_Z (1-\rho^2) \vec{Z}^\trans \vec{Z}.
	\end{align*}
	The log restricted likelihood is then
	\begin{align*}
			\log RL
			&= C - \frac{n-p}{2}\log(\tau_U^{-1} + \tau_\epsilon^{-1}) - \frac{n}{2}\log[\tau_Z^{-1} (1-\rho^2)^{-1}]\\
			&\phantom{==}- \frac{1}{2} \left[
			(\tau_U^{-1} + \tau_\epsilon^{-1})^{-1}
			\left( \vec{Y} - \rho \sqrt{\frac{\tau_Z}{\tau_U}} \vec{Z} \right)^\trans
			\left\{ \mat{I} - \mat{X}(\mat{X}^\trans \mat{X})^{-1} \mat{X}^\trans \right\}
			\left( \vec{Y} - \rho \sqrt{\frac{\tau_Z}{\tau_U}} \vec{Z} \right)
			\right] \\
			&\phantom{==}- \frac{1}{1} \tau_Z (1-\rho^2) \vec{Z}^\trans \vec{Z}, \\
			&= C - \frac{n-p}{2}\log(\tau_U^{-1} + \tau_\epsilon^{-1}) + \frac{n}{2} \log \tau_Z + \frac{n}{2} \log (1-\rho^2) \\
			&\phantom{==}- \frac{1}{2} (\tau_U^{-1} + \tau_\epsilon^{-1})^{-1} \vec{Y}^\trans \left\{ \mat{I} - \mat{X}(\mat{X}^\trans \mat{X})^{-1} \mat{X}^\trans \right\} \vec{Y} \\
			&\phantom{==}+(\tau_U^{-1} + \tau_\epsilon^{-1})^{-1} \rho \sqrt{\frac{\tau_Z}{\tau_U}} \vec{Z}^\trans \left\{ \mat{I} - \mat{X}(\mat{X}^\trans \mat{X})^{-1} \mat{X}^\trans \right\} \vec{Y} \\
			&\phantom{==}-\frac{1}{2} (\tau_U^{-1} + \tau_\epsilon^{-1})^{-1} \rho^2 \frac{\tau_Z}{\tau_U} \vec{Z}^\trans \left\{ \mat{I} - \mat{X}(\mat{X}^\trans \mat{X})^{-1} \mat{X}^\trans \right\} \vec{Z} \\
			&\phantom{==}-\frac{1}{2} \tau_Z (1-\rho^2) \vec{Z}^\trans \vec{Z} \\
			&= C - \frac{n-p}{2}\log(\tau_U^{-1} + \tau_\epsilon^{-1}) + \frac{n}{2} \log \tau_Z + \frac{n}{2} \log (1-\rho^2) \\
			&\phantom{==}- \frac{1}{2} (\tau_U^{-1} + \tau_\epsilon^{-1})^{-1} \vec{Y}^\trans \left\{ \mat{I} - \mat{X}(\mat{X}^\trans \mat{X})^{-1} \mat{X}^\trans \right\} \vec{Y} \\
			&\phantom{==}-\frac{1}{2} \tau_Z (1-\rho^2) \vec{Z}^\trans \vec{Z},
	\end{align*}
	where the last equation holds because $\vec{Z}$ is a column of $\mat{X}$, and therefore $\vec{Z}^\trans \left\{ \mat{I} - \mat{X}(\mat{X}^\trans \mat{X})^{-1} \mat{X}^\trans \right\} = \vec{0}$.
	Treating the restricted likelihood as a function of $\rho$, due to the term $+ \frac{n}{2} \log (1-\rho^2)$, the log restricted likelihood approaches $-\infty$ as $\rho \to \pm 1$, and so all maxima on $\rho \in [-1, 1]$ are in the interior.
	
	Writing $(\tau_U^{-1} + \tau_\epsilon^{-1}) = \sigma^2$ and $\phi = \tau_Z (1-\rho^2)$ we then have
	\begin{align*}
			\log RL
			&= C - \frac{n-p}{2}\log\sigma^{2} - \frac{1}{2}\sigma^{-2} \vec{Y}^\trans \left\{ \mat{I} - \mat{X}(\mat{X}^\trans \mat{X})^{-1} \mat{X}^\trans \right\} \vec{Y} +\frac{n}{2} \log \phi - \frac{1}{2} \phi \vec{Z}^\trans \vec{Z}, \\
			\frac{\partial}{\partial \sigma^2} \log RL
			&= -\frac{n-p}{2\sigma^2} + \frac{1}{2\sigma^4} \vec{Y}^\trans \left\{ \mat{I} - \mat{X}(\mat{X}^\trans \mat{X})^{-1} \mat{X}^\trans \right\} \vec{Y} \\
			\frac{\partial}{\partial \phi} \log RL
			&= \frac{n}{2\phi} - \frac{1}{2}\vec{Z}^\trans \vec{Z}.
	\end{align*}
	Thus $\tau_U$ and $\tau_\epsilon$ are not identifiable, and the global maximum is at $$\sigma^2 = (n-p)^{-1} \vec{Y}^\trans \left\{ \mat{I} - \mat{X}(\mat{X}^\trans \mat{X})^{-1} \mat{X}^\trans \right\} \vec{Y}.$$
	Similarly, $\tau_Z$ and $\rho$ are not identifiable, and the global maximum is at $\phi = n / \vec{Z}^\trans \vec{Z}$.
	As a final result, $\widebar{\vec{\beta}} = (\mat{X}^\trans \mat{V}^{-1} \mat{X})^{-1} \mat{X}^\trans \mat{V}^{-1} (\vec{Y} - \mat{B} \vec{Z}) = (\mat{X}^\trans \mat{X})^{-1} \mat{X}^\trans (\vec{Y} - \rho \sqrt{\frac{\tau_Z}{\tau_U}}\vec{Z})$ is undetermined.
	
	\section{Simulation results when unobserved confounder is correlated with observed confounder}
	\label{app:sim-xucorr}
	
	\begin{table}[H]
		\small
		\caption{Simulation results from 100 data sets of size $n = 300$. The -RS suffix indicates estimators with the restriction $\phi_Z \leq \phi_U$. \vspace{8pt}}
		\label{tab:sim}
		\centering
		\begin{tabularx}{0.6\textwidth}{rd{-2}d{-2}d{-2}d{-2}}
			\toprule
			\multicolumn{1}{c}{Estimator}  & \multicolumn{1}{P{5em}}{Bias} & \multicolumn{1}{P{5em}}{Std. Err.} & \multicolumn{1}{P{5em}}{RMSE} & \multicolumn{1}{P{5em}}{95\% CI \newline Coverage} \\
			\midrule
			Non-spatial       & 0.36 & 0.17 & 0.40 & 0.03 \\
			Spatial       & 0.35 & 0.12 & 0.37 & 0.08 \\
			Spatial-RS    & 0.35 & 0.12 & 0.37 & 0.05 \\
			Affine    & 0.21 & 0.42 & 0.46 & 0.91 \\
			Affine-RS & 0.11 & 0.26 & 0.28 & 0.95 \\
			\bottomrule
		\end{tabularx}
	\end{table}
	
	\section{Simulation results maximum a posteriori estimation}
	\label{app:sim-map}
	
	Here we present some simulation results for linear models for continuous outcomes. Estimation is based on the REML approach presented in the manuscript.
	The exact regularization prior is used because it does not add a substantial computational burden beyond that otherwise necessary for the rest of the MAP estimation procedure.
	
	\begin{table}[H]
		\small
		\caption{Simulation results from 1000 data sets of size $n = 100$. The -RS suffix indicates estimators with the restriction $\phi_Z \leq \phi_U$. \vspace{8pt}}
		\label{tab:sim}
		\centering
		\begin{tabularx}{0.9\textwidth}{rlrd{-2}d{-2}d{-2}d{-2}}
			\toprule
			&
			\multicolumn{1}{l}{Mechanism} & \multicolumn{1}{c}{Estimator}  & \multicolumn{1}{P{5em}}{Bias} & \multicolumn{1}{P{5em}}{Std. Err.} & \multicolumn{1}{P{5em}}{RMSE} & \multicolumn{1}{P{5em}}{95\% CI \newline Coverage} \\
			\midrule
			GM 1 & Independent
			& Non-spatial &  0.00 & 0.19 & 0.19 & 0.94 \\
			& & Spatial & 0.00 & 0.19 & 0.19 & 0.93 \\
			& & Spatial-RS & 0.00 & 0.19 & 0.19 & 0.93 \\
			& & Affine & 0.01 & 0.34 & 0.34 & 0.98 \\
			& & Affine-RS & 0.01 & 0.35 & 0.35 & 0.98 \\
			\addlinespace
			\addlinespace
			GM 2 & Large-scale 
			& Non-spatial & 0.67 & 0.15 & 0.69 & 0.00 \\
			& confounder & Spatial & 0.64 & 0.15 & 0.65 & 0.01 \\
			& & Spatial-RS & 0.63 & 0.14 & 0.64 & 0.01 \\
			& & Affine & 0.49 & 0.47 & 0.68 & 0.81 \\
			& & Affine-RS & 0.24 & 0.36 & 0.43 & 0.96 \\
			\addlinespace
			\addlinespace
			GM 3 & Large-scale 
			& Non-spatial & 0.56 & 0.12 & 0.57 & 0.01 \\
			& exposure & Spatial & 0.55 & 0.12 & 0.57 & 0.01 \\
			& & Spatial-RS & 0.54 & 0.13 & 0.56 & 0.02 \\
			& & Affine & 0.56 & 0.28 & 0.63 & 0.87 \\
			& & Affine-RS & 0.30 & 0.59 & 0.66 & 0.88 \\
			\addlinespace
			\addlinespace
			GM 4 & Same scales
			& Non-spatial & 0.60 & 0.14 & 0.62 & 0.01 \\
			& & Spatial & 0.59 & 0.14 & 0.61 & 0.01 \\
			& & Spatial-RS & 0.58 & 0.14 & 0.59 & 0.02 \\
			& & Affine & 0.55 & 0.37 & 0.66 & 0.85 \\
			& & Affine-RS & 0.31 & 0.47 & 0.56 & 0.92 \\
			\addlinespace
			\addlinespace
			GM 5 & Non-constant
			& Non-spatial & 0.37 & 0.09 & 0.38 & 0.03 \\
			& conditional & Spatial & 0.36 & 0.09 & 0.37 & 0.04 \\
			& correlation & Spatial-RS & 0.36 & 0.09 & 0.37 & 0.04 \\
			& & Affine & 0.21 & 0.32 & 0.38 & 0.95 \\
			& & Affine-RS & 0.09 & 0.27 & 0.28 & 0.98 \\
			\addlinespace
			\addlinespace
			GM 6 & Non-normal
			& Non-spatial & 1.15 & 0.50 & 1.25 & 0.01 \\
			& joint & Spatial & 1.03 & 0.45 & 1.13 & 0.01 \\
			& distribution & Spatial-RS & 1.07 & 0.48 & 1.17 & 0.01 \\
			& & Affine & 0.74 & 0.63 & 0.97 & 0.84 \\
			& & Affine-RS & 0.59 & 0.56 & 0.81 & 0.90 \\
			\bottomrule
		\end{tabularx}
	\end{table}

\clearpage

\bibliographystyle{abbrvnat}
\bibliography{notes}

\end{document}